\documentclass[journal]{IEEEtran}

\usepackage{subfigure}
\usepackage{url}

\usepackage[ruled,vlined,boxed]{algorithm2e}
\usepackage{algorithmic}

\usepackage{import} %
\usepackage{multirow}

   \usepackage[pdftex]{graphicx}

\usepackage{amssymb}
\usepackage{amsmath}
\usepackage{amsthm}
\usepackage{enumitem}

\newtheorem{definition}{Definition}
\newtheorem{theorem}{Theorem}

\usepackage{color}

\newcommand{\descr}[1]{\smallskip\noindent\textbf{#1.}} %
\renewcommand{\paragraph}{\descr}

\newcommand{\A}{\ensuremath \mathcal{A}}
\newcommand{\C}{\ensuremath \mathcal{C}}
\newcommand{\D}{\ensuremath \mathcal{D}}

\newcommand{\Z}{\ensuremath \mathbb{Z}}
\newcommand{\PRF}{\ensuremath \operatorname{PRF}}

\usepackage{pifont}%
\newcommand{\cmark}{\ding{51}}%
\newcommand{\xmark}{\ding{55}}%

\begin{document}

\bstctlcite{IEEEexample:BSTcontrol}

\title{HMOG: New Behavioral Biometric Features for Continuous Authentication of Smartphone Users*}

\author{Zde\v nka Sitov\'a$^\dag$~~
        Jaroslav \v Sed\v enka$^\dag$~~
        Qing Yang$^{\ddag\#}$\thanks{*This paper was published as~\cite{7349202}
        
\smallskip
$^\#$Work done in part while visiting the New York Institute of Technology.
        
\smallskip
$^\flat$Corresponding author. Mailing address: Room HSH 226C, Northern Blvd., NY 11568, USA. Phone: 561-686-1373. Fax: 516-686-7933.

\smallskip
Zde\v nka Sitov\'a and Jaroslav \v Sed\v enka were students at Masaryk University, Czech Republic. This work was done while visiting NYIT.

\smallskip
This work was supported in part by DARPA Active Authentication grant FA8750-13-2-0266 and a 2013 NYIT ISRC grants. The views, findings, recommendations, and conclusions contained herein are those of the authors and should not be interpreted as necessarily representing the official policies or endorsements, either expressed or implied, of the sponsoring agencies or the U.S. Government.

\smallskip
Zhou was supported in part by NSF CAREER grant CNS-1253506.

\smallskip
Computing and storage resources were provided in part under program ``Projects of Large Infrastructure for Research, Development, and Innovations'' (LM2010005).

\smallskip
Copyright (c) 2013 IEEE. Personal use of this material is permitted. However, permission to use this material for any other purposes must be obtained from the IEEE by sending a request to pubs-permissions@ieee.org
}~~
        Ge Peng$^\ddag$~~
        Gang Zhou$^\ddag$~~
        Paolo Gasti$^\dag$~~
        Kiran S.~Balagani$^{\dag \flat}$ \\
        \medskip $^\dag$New York Institute of Technology  --  {\tt \small \{sitovaz, sedenka\}@mail.muni.cz}, {\tt \small \{pgasti, kbalagan\}@nyit.edu}  \\
        $^\ddag$College of William and Mary -- {\tt \small \{qyang, gpeng, gzhou\}@cs.wm.edu}
}

\maketitle

\begin{abstract}
We introduce {\em Hand Movement, Orientation, and Grasp} (HMOG), a set of behavioral features to continuously authenticate smartphone users. HMOG features unobtrusively capture subtle micro-movement and orientation dynamics resulting from how a user grasps, holds, and taps on the smartphone. We evaluated authentication and biometric key generation (BKG) performance of HMOG features on data collected from 100 subjects typing on a virtual keyboard. Data was collected under two conditions: sitting and walking. We achieved authentication EERs as low as 7.16\% (walking) and 10.05\% (sitting) when we combined HMOG, tap, and keystroke features. We performed experiments to investigate why HMOG features perform well during walking. Our results suggest that this is due to the ability of HMOG features to capture distinctive body movements caused by walking, in addition to the hand-movement dynamics from taps. With BKG, we achieved EERs of 15.1\% using HMOG combined with taps. In comparison, BKG using tap, key hold, and swipe features had EERs between 25.7\% and 34.2\%. We also analyzed the energy consumption of HMOG feature extraction and computation.
Our analysis shows that HMOG features extracted at 16Hz sensor sampling rate incurred a minor overhead of 7.9\% without sacrificing authentication accuracy.
Two points distinguish our work from current literature: 1) we present the results of a comprehensive evaluation of three types of features (HMOG, keystroke, and tap) and their combinations under the same experimental conditions; and 2) we analyze the features from three perspectives (authentication, BKG, and energy consumption on smartphones).

\end{abstract}


\IEEEpeerreviewmaketitle

\section{Introduction}

Currently, popular smartphone authentication mechanisms such as PINs, graphical passwords, and fingerprint scans offer limited security. They are susceptible to guessing~\cite{pin_distribution} (or spoofing~\cite{iphone_fingerprint_spoof} in the case of fingerprint scans), and to side channel attacks such as smudge~\cite{AvivGMBS10}, reflection~\cite{XuH0MF13}, and video capture~\cite{ShuklaKSP14} attacks.
Additionally, a fundamental limitation of PINs, passwords, and fingerprint scans is that they are well-suited for one-time authentication, 
and therefore are commonly used to authenticate users at login. This renders them ineffective when the smartphone is accessed by an adversary after login. {\em Continuous}  {or \em active} authentication addresses these challenges by frequently and unobtrusively authenticating the user via behavioral biometric signals, such as touchscreen interactions~\cite{frank2013}, hand movements and gait~\cite{derawi2010, bo2013}, voice~\cite{LuBPKL11}, and phone location~\cite{shi2010}.

In this paper, we present Hand Movement, Orientation, and Grasp (HMOG), a new set of behavioral biometric features for continuous authentication of smartphone users. HMOG uses accelerometer, gyroscope, and magnetometer readings to unobtrusively capture subtle hand micro-movements and orientation patterns generated when a user taps on the screen. 

HMOG features are founded upon two core building blocks of human prehension~\cite{Napier56}: {\em stability grasp}, which provides stability to the object being held; and {\em precision grasp}, which involves precision-demanding tasks such as tapping a target. We view the act of \textit{holding a phone} as a stability grasp and the act of \textit{touching targets on the touchscreen} as a precision grasp. We hypothesize that the way in which a user ``distributes'' or ``shares'' stability and precision grasps while interacting with the smartphone results in distinctive movement and orientation behavior. The rationale for our hypothesis comes from the following two bodies of research.

First, there is evidence~(see \cite{Karlson06,Azenkot2012,Wobbrock2008}) that users have postural preferences for interacting with hand-held devices such as smartphones. Depending upon the postural preference, it is possible that the user can have her own way of achieving stability and precision---for example, the user can achieve both stability and precision with one hand if the postural preference involves holding and tapping the phone with the same hand; or distribute stability and precision between both hands, if the posture involves using both hands for holding and tapping; or achieve stability with one hand and precision with the other. %

Second, studies in ergonomics, biokinetics, and human-computer interaction have reported that handgrip strength strongly correlates with an individual's physiological and somatic traits like hand length, handedness, age, gender, height, body mass, and musculature (see, e.g.,  \cite{fiebert1998relationship,chatterjee1991comparison,kim2006hand}). If the micro-movements caused by tapping reflect an individual's handgrip strength, then the distinctiveness of HMOG may have roots, at least in part, in an individual's distinctive physiological and somatic traits.

Motivated by the above, we designed 96 HMOG features and evaluated their {\em continuous user authentication} and {\em biometric key generation} performance during typing. Because walking has been shown to affect typing performance~\cite{MizobuchiCN05}, we evaluated HMOG under both walking and sitting conditions.

\subsection{Contributions and Novelty of This Work}

\paragraph{New HMOG Features for Continuous Authentication} 
We propose two types of HMOG features: {\em resistance features}, which measure the micro-movements of the phone in response to the forces exerted by a tap gesture; and {\em stability features}, which measure how quickly the perturbations in movement and orientation, caused by tap forces, dissipate. %
Our extensive evaluation of HMOG features on a dataset  of 100 users\footnote{We made the dataset available at \url{http://www.cs.wm.edu/~qyang/hmog.html}. We also described the data and its release in \cite{qingPoster}.} who typed on the smartphone led to the following findings: (1)~HMOG features extracted from accelerometer and gyroscope signals outperformed HMOG features from magnetometer;
(2) Augmenting HMOG features with tap characteristics (e.g., tap duration and contact size) lowered equal error rates (EERs): from 14.34\% to 11.41\% for sitting, and from 14.73\% to 8.53\% for walking.
This shows that combining tap information with HMOG features considerably improves authentication performance; and (3) HMOG features complement tap and keystroke dynamics features, especially for low authentication latencies at which tap and keystroke dynamics features fare poorly. For example, for 20-second authentication latency, adding HMOG to tap and keystroke dynamics features reduced the equal error rate from 17.93\% to 11.74\% for walking and from 19.11\% to 15.25\% for sitting.

\paragraph{Empirical Investigation Into Why HMOG Authentication Performs Well During Walking}  HMOG features achieved lower authentication errors (13.62\% EER) for walking compared to sitting (19.67\% EER). We investigated why HMOG had a superior performance during walking by comparing the performance of HMOG features {\em during} taps and {\em between} taps (i.e., the segments of the sensor signal that lie between taps). Our results suggest that the higher authentication performance %
during walking can be attributed to the ability of HMOG features to capture distinctive 
movements caused by walking in addition to micro-movements caused by taps. %

\paragraph{BKG with HMOG Features} 
BKG is closely related to authentication, but has a different objective: to 
provide cryptographic access control to sensitive data on the smartphone. We 
believe that designing a secure BKG scheme on smartphones is very important, 
because the adversary is usually assumed to have physical access to the device, 
and therefore cryptographic keys must not be stored on the smartphone's memory---but rather generated from biometric signals and/or passwords.

To our knowledge, we are the first to evaluate BKG on smartphones. We instantiated BKG using normalized generalized Reed-Solomon codes in Lee metric. (See section~\ref{sec:bkg} for formulation and evaluation.) We compared BKG on HMOG to BKG on tap, key hold, and swipe features under two metrics: equal error rate (EER) and guessing distance. 

Our results on BKG can be summarized as follows: we achieved lower EERs with HMOG features compared to key hold, tap, and swipe features in both walking and sitting conditions. For walking, EER of HMOG-based BKG was 17\%, vs. 29\% with key hold and 28\% with tap features. By combining HMOG and tap features, we achieved 15.1\% EER. For sitting, we obtained an EER of 23\% with HMOG features, 26\% with tap features, and 20.1\% by combining both. In contrast, we obtained 34\% EER with swipes. HMOG features also provided higher guessing distance (i.e., 2.9 for walking, and 2.8 for sitting) than all other features extracted from our dataset (1.9 for taps and for key holds in walking and 1.6 for taps in sitting conditions).

\paragraph{Energy Consumption Analysis of HMOG Features}
Because smartphones are energy constrained, it is crucial that a continuous user authentication method consumes as little energy as possible, while maintaining the desired level of authentication performance. To evaluate the feasibility of HMOG features for continuous authentication on smartphones, we measured the energy consumption of accelerometer and gyroscope, sampled at 100Hz, 50Hz, 16Hz and 5Hz. We then measured the energy required for HMOG feature computation from sensor signals, and reported the tradeoff between energy consumption and EER.

Our analysis shows that a balance between authentication performance and energy overhead can be achieved by sampling HMOG features at 16Hz. The energy overhead with 16Hz is 7.9\%, compared to 20.5\% with 100Hz sampling rate, but comes with minor increase (ranging from 0.4\% to 1.8\%) in EERs. However, by further reducing the sampling rate to 5Hz, we observed a significant increase in EER (11.0\% to 14.1\%). %

\subsection{Organization}
We present the description of HMOG features in Section \ref{sectionFeatures}, and details on our dataset in Section \ref{sectionDataset}. In sections \ref{sectionExperiments} and \ref{sectionResults}, we describe the authentication experiments and present results. We introduce and evaluate BKG on HMOG in Section \ref{sec:bkg}. We analyze the energy consumption of HMOG features in Section \ref{sec:power}. In Section \ref{relatedResearch}, we review related research.  We conclude in Section \ref{sec:conclusion}.

\section{Description of HMOG features} \label{sectionFeatures}

We define two types of HMOG features: grasp {\em resistance} and grasp {\em stability}. These features are computed from data collected using three sensors: accelerometer, gyroscope, and magnetometer. Because HMOG features aim to capture the subtle micro-movements and orientation patterns of a user while tapping on the screen, we extract HMOG features from signals collected {\em during} or {\em close to} tap events.  
Computation of grasp stability and resistance features is discussed next.

\subsection{Grasp Resistance Features} 
Grasp resistance features measure the resistance of a hand grasp to the forces (or pressures) exerted by touch/gesture events. We quantify resistance as the change, or perturbation, in movement (using readings from accelerometer), orientation (from gyroscope) and magnetic field (from magnetometer), caused by a tap event.

\begin{table}[htbp]  \caption{Notation.}\label{notation-1}

  \centering \resizebox{.5\textwidth}{!}{
  \begin{tabular}{@{} |c|l| @{}}
    \hline
Sensor & Accelerometer or Gyroscope or Magnetometer. \\\hline   
 $X$, $Y$, $Z$ & Time series of sensor readings in $x$, $y$, and $z$ axes respectively\\ \hline
 $Z_1, \ldots, Z_n$ & Individual sensor readings in $z$ axis collected at time \\ &$t_1, \ldots, t_n$ respectively \\\hline
    $M$ & Time series of magnitude of sensor reading, where each \\ &element $M_i$ is computed as $\sqrt{(X_i^2+Y_i^2+Z_i^2)}$ \\ \hline
    $t_{start}$ & Start time of a tap event \\ \hline
    $t_{end}$ & End time of a tap event \\ \hline
    $t_{max\_in\_tap}$ & Time between $t_{start}$ and $t_{end}$ at which the reading from a \\&sensor   reaches its highest value \\ \hline
    $t_{min}$ & Time when stability is achieved after the tap event has ended \\ \hline
    $t_{before\_center}$ & Center of the 100 ms window before a tap \\ \hline
    $t_{after\_center}$ & Center of the 100 ms window after a tap \\ \hline
    \texttt{avg100msBefore} & Average of sensor readings in a 100~ms  window before start \\ \texttt{avg100msAfter}&  time and after end time, respectively\\ \hline
    \texttt{avgTap} & Average of readings during tap events \\ \hline
    $t_{min}$ & Time when stability is achieved after the tap event has ended\\
    \hline
  \end{tabular}}
\end{table}

We extracted five grasp resistance features from accelerometer, gyroscope, and magnetometer, over four dimensions (magnitude, $x$, $y$, and $z$ axes), leading to $5 \times 3 \times 4 = 60$ features. For simplicity of exposition, we describe grasp resistance features only on the $z$ axis. We also extracted the same features from $X$, $Y$, and $M$. Our notation is summarized in Table~\ref{notation-1}. Figure~\ref{fig:drawingStabilityFeatures} illustrates variables used in features 3 through 5.
\begin{enumerate}
\item Mean of $Z$ during taps. %
\item Standard deviation of $Z$ during taps.
\item Difference in $Z$  readings before and after a tap event. Let \texttt{avg100msBefore} be the average of $Z$ readings in a 100~ms  window before tap start time, and \texttt{avg100msAfter} be the average of $Z$ readings in a 100~ms  window after tap end time. We calculated this feature as the difference between \texttt{avg100msAfter} and \texttt{avg100msBefore}.
\item Net  change in $Z$ readings caused by a tap. Let  \texttt{avgTap} be the average of $Z$ readings during a tap event. We calculate this feature as  \texttt{avgTap - avg100msBefore}.
\item Maximum change in $Z$ readings caused by a tap. Let \texttt{maxTap} be the maximum $Z$ reading during a tap event. This feature is calculated as 
 \texttt{maxTap - avg100msBefore}.
\end{enumerate}

\subsection{Grasp Stability Features} 
Stability features quantify how quickly the perturbations caused by a finger-force from a tap event disappear after the tap event is complete.
We compute grasp stability features as follows:
(Figure~\ref{fig:drawingStabilityFeatures} illustrates variables used in the features.)

\begin{figure}[tb]
\includegraphics[width=0.96\linewidth]{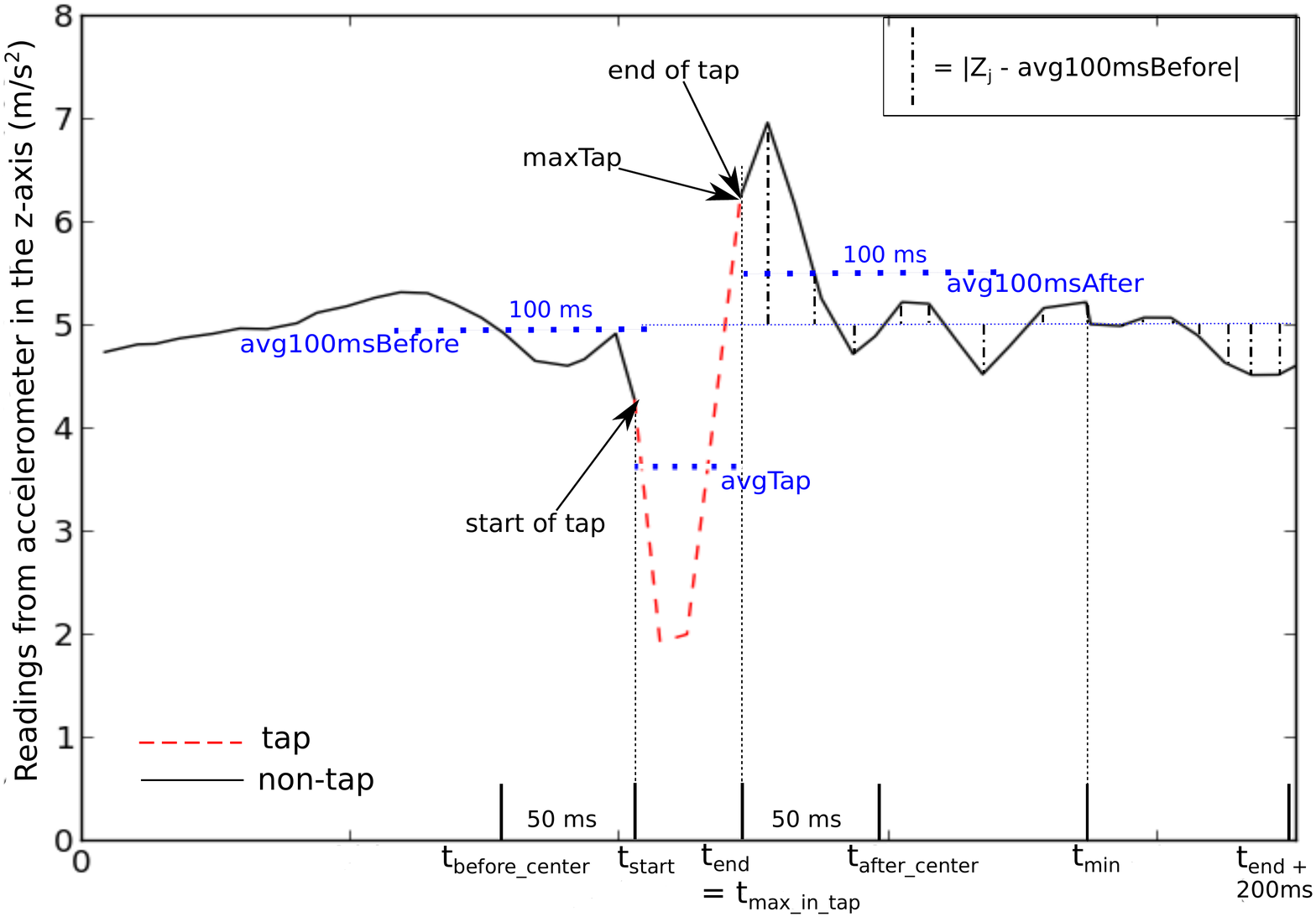}
\caption[]{Illustration of key variables for computing grasp {\em resistance} features 3-5 and grasp {\em stability} features 1-3.}
\label{fig:drawingStabilityFeatures}
\end{figure}

\begin{enumerate}
\item Time duration to achieve movement and orientation stability after a 
tap event. Let $t_{end}$ denote the end time of the tap event, and $t_{min}$ the 
time when stability is achieved after the tap event has ended, computed as shown in Algorithm~\ref{alg:tmin}. This feature is 
calculated as $t_{min} - t_{end}$.
\begin{algorithm}[h!]
\footnotesize
\caption{Computation of $t_{min}$ on $Z$ readings}
\SetKwInOut{Input}{input}\SetKwInOut{Output}{output}
\Input{{\tt avg100msBefore}, $t_1, \ldots, t_n$ (timestamps between $t_{end}$ and $t_{end}+200$~ms, and $Z_1, \ldots, Z_n$ ($Z$ readings at $t_1, \ldots, t_n$)}
\Output{$t_{min}$ (timestamp between $t_{end}$ and $t_{end}+200$~ms at which the sensor reading are closest to those measured before the tap.)}
\begin{algorithmic}[1]
\FOR{$i=1\ldots n$} 
\STATE{{\tt avgDiffs}$[i] = \frac{\sum_{j=i}^n(|Z_j-{\tt avg100msBefore}|)}{n-i+1}$}
\ENDFOR
\STATE{$min = \mathrm{argmin}_i({\tt avgDiffs}[i])$ {\tt //$min$ is the index at\\ which avgDiffs has its minimum value}}
\STATE{return $t_{min}$}
\end{algorithmic}
\label{alg:tmin}
\end{algorithm}

\item {Normalized time duration for mean sensor value to change from before tap to after tap event}, calculated as:
\[\resizebox{0.89\linewidth}{!}{$
\Delta_{\mathit{duration}} = \frac{t_{\mathit{after\_center}} - t_{\mathit{before\_center}}}{\texttt{avg100msAfter} - \texttt{avg100msBefore}}
$}
\]
where $t_{\mathit{after\_center}}$ is the center of the 100ms window after a tap event, and $t_{\mathit{before\_center}}$ is the center of the 100ms window before the tap event.

\item Normalized time duration for mean sensor values to change from \texttt{maxTap} to \texttt{avg100msAfter} in response to a tap event, calculated as:
\[\resizebox{0.8\linewidth}{!}{$
\Delta_{\mathit{max\_to\_avg}} = \frac{t_{after\_center} - t_{{max\_in\_tap}}}{\texttt{avg100msAfter} -\texttt{maxTap}}
$}
\]
where \texttt{maxTap} is the maximum sensor value during a tap, and $t_{{max\_in\_tap}}$ is the time when this value occurred.
\end{enumerate}
We extracted the above three grasp stability features for three sensors and four types of sensor readings ($X$, $Y$, $Z$ and $M$), for a total of $3 \times 3 \times 4 = 36$ features. 

\medskip
Complexity of computing HMOG features is linear (O$(n)$) in the sampling frequency, except for Grasp Stability Feature 1, which is quadratic (O$(n^2)$).

\section{Dataset} \label{sectionDataset}
To evaluate HMOG features, we used sensor data collected from 100 smartphone users (53 male, 47 female) during eight free text typing sessions~\cite{qingPoster}.\footnote{Our dataset is available at \url{ http://www.cs.wm.edu/~qyang/hmog.html}} Users answered three questions per session, typing at least 250 characters for each answer. In four sessions, the users typed while sitting. In another four sessions, users typed while walking in a controlled environment.

For each user, we collected an average of 1193 taps per session (standard deviation: 303) and 1019 key presses (standard deviation: 258). The average duration of a session was 11.6 minutes, with a standard deviation of 4.6 minutes.
Data was collected using the same smartphone model (Samsung Galaxy S4). We used a total of ten Samsung Galaxy S4. Data was collected over multiple days, and the same user might have received a different device during each visit.

We recorded accelerometer, gyroscope and magnetometer sensor readings (sampling 
rate 100~Hz) as well as raw touch data collected from the touchscreen, touch 
gestures (e.g., tap, scale, scroll, and fling), key press, and key release 
latencies on the virtual keyboard. Due to security concerns, Android OS forbids 
third-party applications to access touch and key press data generated on the 
virtual keyboard. Therefore, we designed a virtual keyboard for data collection 
that mimicked the look, feel, and functionality of default Android keyboard, 
including the autocorrect and autocomplete options, which the users were free to use.

During data collection users were allowed to choose the orientation of the 
smartphone (i.e., landscape or portrait). Because less than 20 users typed in landscape orientation, we performed all authentication experiments with data collected in portrait mode.

\section{Evaluation of HMOG Features} \label{sectionExperiments}

\begin{figure}[htbp]
   \centering
   \includegraphics[width=0.8\linewidth]{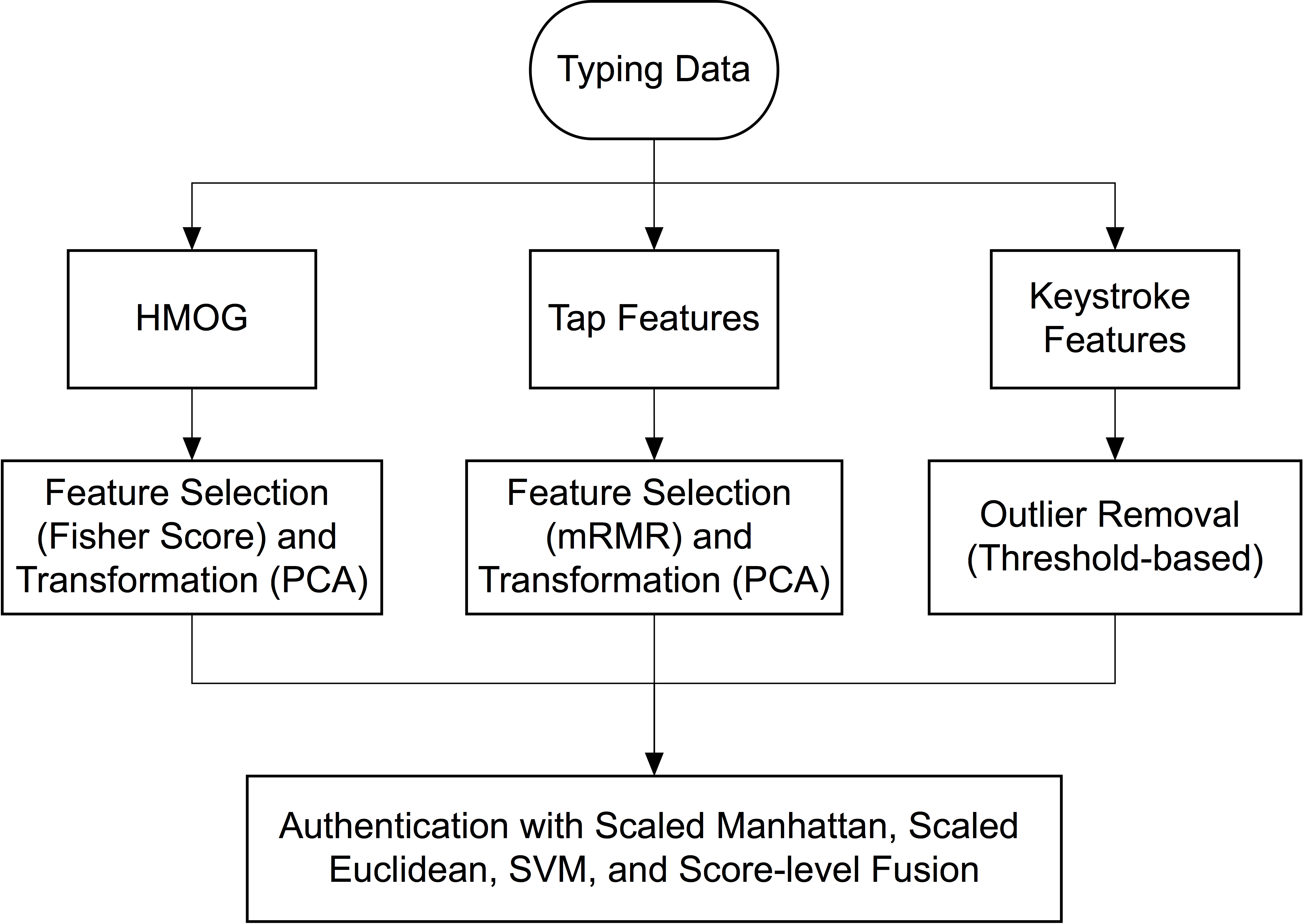} %
   \caption{Flow-diagram depicting our experiment workflow.}
   \label{fig:flowchart}
\end{figure}

Our experiment workflow involves: (1) computing features from data collected during typing; (2) performing feature selection; (3) performing feature transformation (PCA); (4) performing  outlier removal; and (5) performing authentication using  Scaled Manhattan, Scaled Euclidean, SVM verifiers, and score-level fusion. 
Figure~\ref{fig:flowchart} summarizes the experiment workflow.

\subsection{Design of Authentication Experiments} \label{authenticationExperiments}

\paragraph{1-Class Verifiers} We performed verification experiments using three verifiers~\cite{maxion2009}: scaled Manhattan (SM), scaled Euclidian (SE), and 1-class SVM. (Henceforth, we use ``SVM'' to refer to ``1-class SVM''.) We chose these verifiers because previous work on behavioral authentication has shown that they perform well. For instance, SM and SVM were top performers in a study on keystroke authentication of desktop users by Maxion et al.~\cite{maxion2009}. SVM performed well in experiments on touch-based authentication of smartphone users by Serwadda et al.~\cite{serwadda2013}. SE is a popular verifier in biometrics (see for example~\cite{blanton2011,govindarajan2013}).

Parameter tuning was not required for SM and SE. However, for SVM~\cite{libSVM}, we used RBF kernel and performed a grid search to find the parameters (for $\gamma$, we searched through $2^{-13}$, $2^{-11}$,  $2^{-9}$, \dots, $2^{13}$; and for $\nu$, we searched through 0.01, 0.03, 0.05, 0.1, 0.15 and 0.5). We used cross-validation to choose the parameter values (see Section~\ref{featurePreparation}).
 
We did not include 2-class verifiers in our evaluation. To train a 2-class verifier, in addition to data from smartphone owner, biometric data from other users (non-owners) is required. Because sharing of biometric information between smartphone users leads to privacy concerns, we believe that 1-class verifiers are more suitable for smartphone authentication. (A similar argument was made in~\cite{shen2013user}.)

\paragraph{Training and Testing}  For experiments in sitting and walking conditions, we used the first two sessions for training and the remaining two for testing. We extracted HMOG features during each tap. Thus, each training/testing vector corresponded to one tap. With keystroke dynamics features, each training/testing vector corresponded to one key press on the virtual keyboard.%

For SM and SE, the template consisted of the feature-wise average of all training vectors. We used user-wise standard deviations for each feature for scaling. We used all training vectors to construct the template (hypersphere) with SVM.
Users with less than 80 training vectors were discarded from authentication. As a consequence, ten users failed to enroll (and were not included in our experiments).

We created authentication vectors by averaging test vectors sampled during $t$-seconds scan. We report results for authentication scans of $t$ = 20, 40, 60, 80, 100, 120 and 140 seconds. We chose these scan lengths to cover both low and higher authentication latencies. Our preliminary experiments showed that for scans longer than 140 seconds, there is minimal improvement in authentication performance. 

\paragraph{Quantifying Authentication Performance}
We generated two types of scores, genuine (authentication vector was matched against template of the same user) and zero-effort impostor (authentication vector of one user was matched against the template of another). We used population equal error rate (EER) to measure the authentication performance.

\subsection{Comparing HMOG to Other Feature Sets}
We compared the authentication performance of HMOG features with touchscreen tap and keystroke dynamic features (key hold and digraph latencies).

\paragraph{Touchscreen Features from Tap Events} 
We extracted 11 commonly used touchscreen-based features for tap events (see~Table \ref{tab:relatedResearch}). Some papers (e.g.,  \cite{serwadda2013} and \cite{frank2013}) defined these features for swipes, while we extracted them from {\em taps} due to very low availability of swipes during typing, and to provide a more meaningful comparison with HMOG features, which are collected during taps. The features we extracted are: 
\begin{itemize}
\item Duration of the tap
\item Contact size features: mean, median, standard deviation, 1st, 2nd and 3rd quartile, first contact size of a tap, minimum and maximum of the contact size during the tap (9 features) 
\item Velocity (in pixels per second) between two consecutive \textit{press} events belonging to two consecutive taps.
\end{itemize}

\paragraph{Key Hold Features}
Key hold latency is the down-up time between press and release of a key. We used 89 key hold features, each corresponding to a key on the virtual keyboard. 

\paragraph{Digraph Features}
Digraph latency is the down-down time between two consecutive key presses. We used digraph features for combinations of the 35 most common keys in our dataset.\footnote{All 26 alphabetic keys, 5 keyboard switches (shift, switch between numerical and alphabetical keyboard, delete, done, return) and 4 special characters (space, dot, comma and apostrophe). The availability of other keys in our training data was extremely low ($< 1$ on average per user).} Thus we have $35^2 = 1225$ digraph features.%

\paragraph{Score-level Fusion}
To determine whether HMOG features complement existing feature sets, we combined tap, key hold, digraph and HMOG features using weighted sum score-level fusion. We chose this method because it is simple to implement, and has been shown to perform well in biometrics~\cite{ross2003}. We used the technique of Locklear et al.~\cite{locklear2014} to ensure that weights sum to one and proportion of weights is preserved when scores from some feature sets were missing (e.g. due to lack of accelerometer data). We used grid-search to find the weights which led to the best authentication performance.

\subsection{Feature Selection, Preprocessing, and Transformation} \label{featurePreparation}
To improve  authentication performance, we performed feature selection, feature transformation with Principal Component Analysis (PCA), and outlier removal. %

\paragraph{Parameter selection}
We used 10-fold cross-validation (10-CV) on training data to choose feature selection method (mRMR or Fisher score ranking), as well as to set the parameters for feature selection, PCA, and SVM. The parameters are presented in Table \ref{cvparams}. We evaluated all parameters independently for each combination of feature set, verifier, authentication scan-length and body-motion condition. 

\begin{table}
\caption{Parameters evaluated using cross-validation} \centering
\begin{tabular}{ l  l  }
\hline
Method & Parameter \\ \hline
Fisher score & Percentage of the sum of all Fisher scores \\
mRMR & Threshold on the mRMR score\\
PCA & Percentage of total variance \\
SVM & $\gamma$, $\nu$ \\
\hline
\end{tabular}
\label{cvparams}
\end{table}

For each set of parameter values, 10-CV yielded ten EERs, which we averaged to get an estimate of the EER corresponding to that set of parameter values. We then selected parameter values which had the lowest (average) EER. For 10-CV experiments involving 20- to 140-second scan lengths, the sets of parameter values that led to the lowest EERs were not always identical. In this case, we took a majority vote to select the most common parameter values.

\paragraph{Feature Selection} 
During training, we evaluated two feature selection methods: Fisher score ranking~\cite{duda2001pattern}, 
and minimum-Redundancy Maximum-Relevance (mRMR)~\cite{mRMR}. 
Our preliminary experiments showed that Fisher score performed better for HMOG features, while mRMR performed well with tap features. With key hold and digraph features, the best performing feature set contained all the features.

Fisher score ranking 
was computed independently for each HMOG feature as the ratio of {\em between-user} 
to {\em within-user} variance. (Higher Fisher score suggests higher 
discriminability of the corresponding feature.) Using 10-CV, we tested feature 
subsets whose sum of Fisher scores accounted for 80\% to 100\% of the 
sum of Fisher scores of all features.  

We selected HMOG features for each verifier separately. The following parameters 
for Fisher score ranking provided the best authentication results: 82\% (17 features) for SM
during sitting; 81\% (13 features) for SM and SVM during walking; and 80\% (16 features) for SVM during sitting. For SE, we achieved lowest EER by including resistance features only, compared to the feature subset obtained from feature selection. %
Figure~\ref{fig:fisherFeatures} reports the ranking of the features during sitting (\ref{fig:fisherFeaturesSit}) and walking (\ref{fig:fisherFeaturesWalk}).

\begin{figure}[tb] 
\subfigure[Fisher scores of HMOG features during sitting.]{\label{fig:fisherFeaturesSit}
\includegraphics[width=1\linewidth]{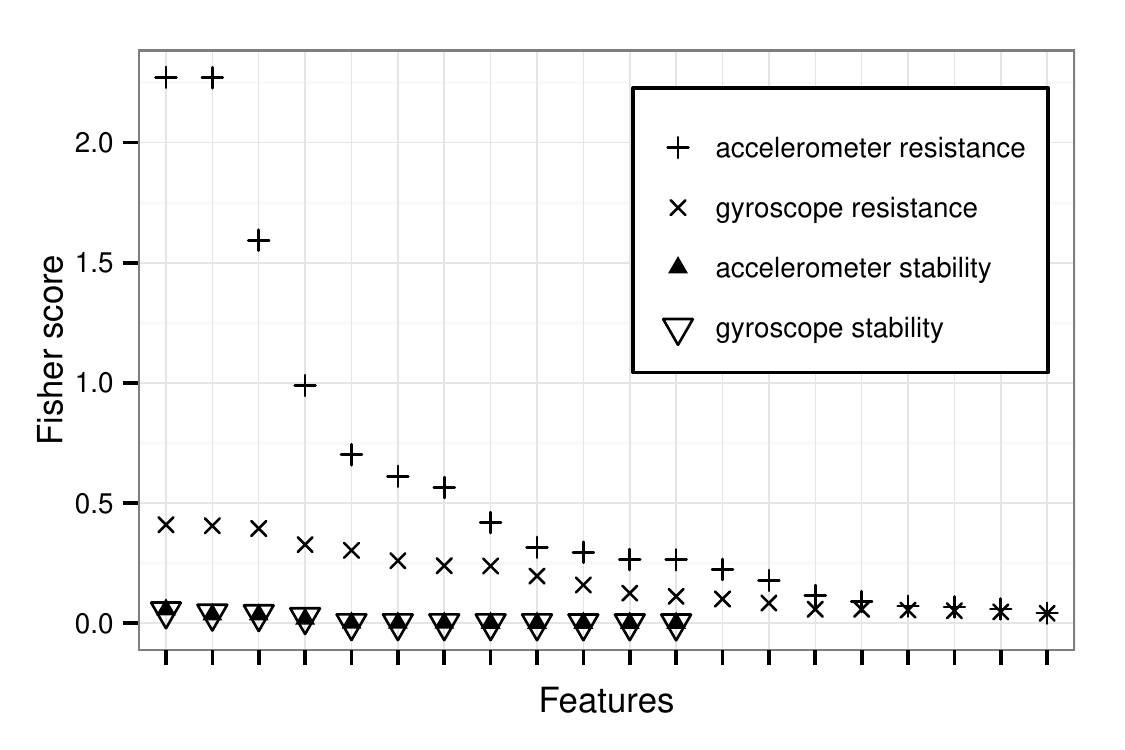} 
}

\subfigure[Fisher scores of HMOG features during walking.]{\label{fig:fisherFeaturesWalk}
\includegraphics[width=1\linewidth]{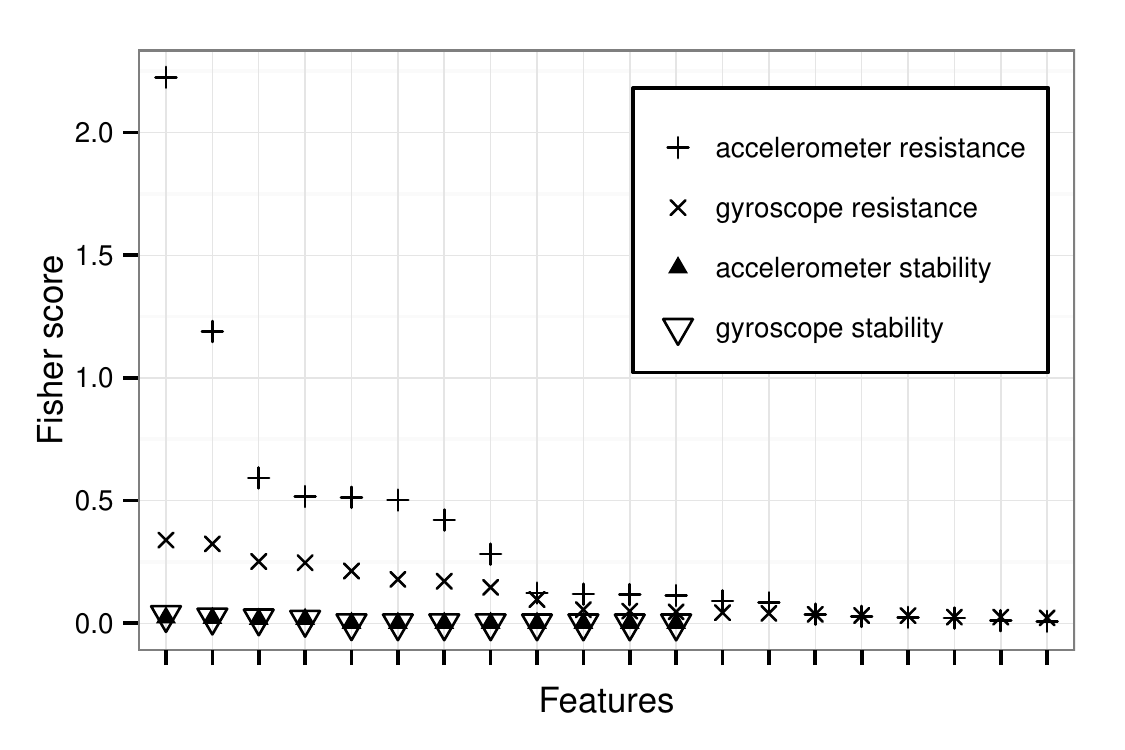} 
}

\caption[]{HMOG features extracted from accelerometer and gyroscope, sorted by Fisher score computed from training data. Higher scores correspond to features with higher discriminative power. 
Magnetometer features (not shown) ranked below accelerometer and gyroscope features.
}
 \label{fig:fisherFeatures}
\end{figure}

For tap features, with SM verifier we achieved the best results  with 3 features chosen by mRMR (threshold 0) and for SE and SVM with 2 features (threshold 0.1). The best three features according to mRMR are (in this order): duration of the tap; mean of contact size; and velocity between two consecutive down events. 

\paragraph{Outlier Removal}
For HMOG and tap templates, we evaluated the interquartile outlier removal (i.e., different subsets of the values from the first and fourth quartile are removed). Experiments with SM verifier showed that outlier removal does not improve authentication accuracy, so we did not consider it further in our experiments. 

For key hold and digraph latencies, using only outlier removal and not performing feature selection or transformation led to the best results. Outlier removal was done using two parameters: (1) latencies longer than $l$ ms were discarded and (2) if a feature occurs less than $m$ times in a user's template, the feature was discarded). 
The values evaluated for $l$ were 100, 200, 300, 400, 500 and 1000 for key hold and 200, 350, 500, 650 and 800 for digraph.
 For $m$, we experimented with 2, 5, 10, 15, 20, 40 and 60. 
The best $l$ value was 200 for key hold, and between 350--500 for digraph; the best $m$ value was between 2--60 for key hold and between 2--5 for digraph.

\paragraph{Feature Transformation} 
We used PCA to transform original features into principal components, that were subsequently used in authentication experiments. Our motivation for using PCA are: (1) to remove correlation between features to meet the assumptions in SE and SM, and (2) to reduce dimensionality by using only those principal components, which explain most of the variance. We performed PCA under two settings: (1) on all features (except magnetometer features, which performed poorly), and  (2) on a subset of features selected using Fisher score and mRMR. We performed 10-CV experiments with components explaining 90\%, 95\%, 98\%, and 100\% of total variance, to set the threshold for dimensionality reduction. PCA improved EER for HMOG features with SE when performed on resistance features, and for SVM during sitting when performed on features selected using Fisher score. %
PCA performed on all tap features improved results with SM and SE. %

\section{Authentication Results} \label{sectionResults}

In this section, we report authentication performance of HMOG features. We compare the performance of HMOG with keystroke and tap features and report results with fusion.\footnote{See Supplement for the number of genuine and impostor scores used for calculating EERs in this paper.} Finally, we present our findings on why HMOG features achieve lower EERs during walking. %

\subsection{Performance of HMOG Features}

\begin{table*}[!htbp]
\caption{Summary of Lowest EERs achieved using only HMOG features.}%
\centering
\begin{tabular}{| c | c | c | c | c |} %
\hline
Verifier 	    			   & Best Performing 	 Features                & Sensors                                &   Sitting             &     Walking   \\ \hline  %
Scaled Manhattan        & With Fisher Score Ranking & Accelerometer + Gyroscope	& 19.67\% & 13.62\%  \\ \hline %
Scaled Euclidean         &  With PCA, no Feature Selection	& Accelerometer + Gyroscope			& 25\% & 15.31\% \\ \hline %
1-Class SVM  		   & With Fisher Score Ranking; with PCA for sitting, without PCA for walking & Accelerometer + Gyroscope	& 27.45\% & 15.71\% \\ \hline %
\end{tabular}
\label{tab:individualperformance}
\vspace{12pt}
\caption{Summary of Lowest EERs achieved with score-level fusion of HMOG,Tap, and Keystroke Dynamics (KD) Features.}
\centering
\begin{tabular}{| c | c | c | }%
\hline
Score-Level Fusion with SM Verifier 			    		& Sitting          	 	 & Walking \\ \hline %
HMOG, Tap, and Keystroke Dynamics		& \textbf{10.05}\% 		 & \textbf{7.16}\%  \\ \hline %
HMOG and Tap  			    		& 11.41\%      	 	 & 8.53\%  \\ \hline %
Tap and KD				   		& 11.02\%			 & 10.79\% \\ \hline %
\end{tabular}
\label{tab:fusionperformance}
\end{table*}

HMOG features extracted from both accelerometer and gyroscope outperformed those extracted from individual sensors. %
HMOG features from magnetometer performed consistently worse than accelerometer and gyroscope features with all verifiers, in both sitting and walking conditions. Combining magnetometer features with features from accelerometer and gyroscope did not improve performance.

Resistance features outperformed stability features in both walking and sitting conditions (and also had higher Fisher score, see Figure~\ref{fig:fisherFeatures}). This suggests that the ability of resistance features to discriminate between users is higher than that of stability features. In fact, feature selection on HMOG with 10-CV  resulted in selecting resistance features only.
In some cases, using PCA after feature selection further lowered EERs. Table \ref{tab:individualperformance} summarizes the sensors and feature selection/transformation that led to the lowest EERs.

In Figure~\ref{fig:hmogAllVerifiersSitWalk}, we show the EERs of all verifiers under sitting and walking conditions, when the authentication scans varied between 20 and 140 seconds. Among the three verifiers, SM overall had lower EERs for both sitting and walking conditions and therefore we present the results only with SM hereafter. %

\begin{figure}[t]
\includegraphics[width=1\linewidth]{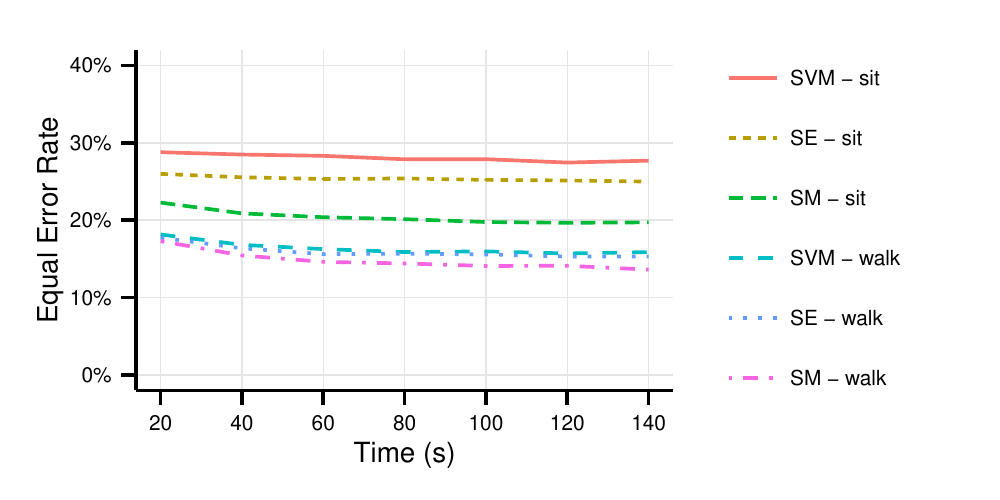}
\caption[]{Comparison of HMOG features in sitting and walking conditions for three verifiers. The reported EERs are with PCA for SE, and for SVM-sitting; and without PCA for SM and SVM-walking. $X$-axis shows authentication time in seconds.}
\label{fig:hmogAllVerifiersSitWalk}
\end{figure}

\begin{figure}[t]
\includegraphics[width=1\linewidth]{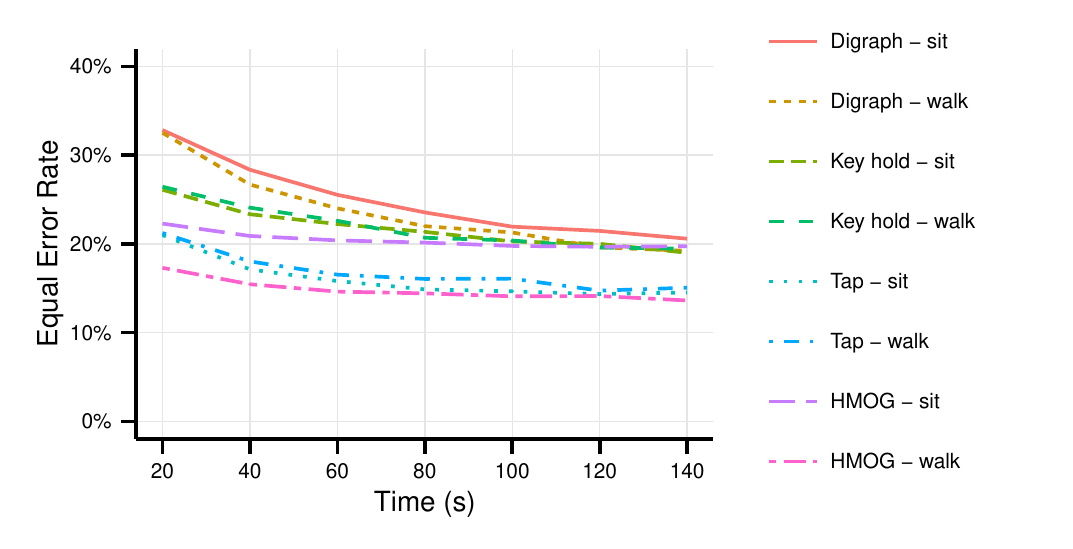}
\caption[]{Comparison of EERs of HMOG with keystroke dynamics (i.e., key hold and digraph) and tap features with SM verifier. $X$-axis shows authentication time in seconds.} %
\label{fig:walkVSsitAllSM}
\end{figure}

\paragraph{Comparison of HMOG with Keystroke Dynamics and Tap Features}
Tap features and HMOG features in walking condition performed better than keystroke dynamics features; HMOG in sitting outperforms keystroke dynamics for shorter scans and is comparable for longer scans (see Figure \ref{fig:walkVSsitAllSM}). 

HMOG features outperformed tap features in walking condition, while tap outperformed HMOG in sitting. 
The performance of tap and keystroke dynamics features did not change significantly between sitting and walking. However, the performance of HMOG improved considerably (up to 6.11\%) during walking. %

\begin{figure}[t]
 \centering
\subfigure[Sitting]
{
\includegraphics[width=1.1\linewidth]{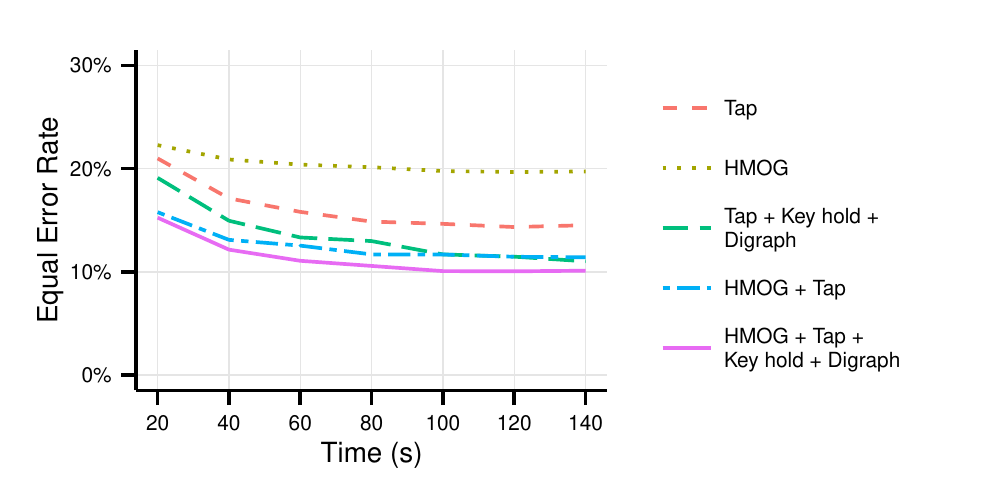}
\label{fig:sitFusionSM}}
\subfigure[Walking]
{
\includegraphics[width=1.1\linewidth]{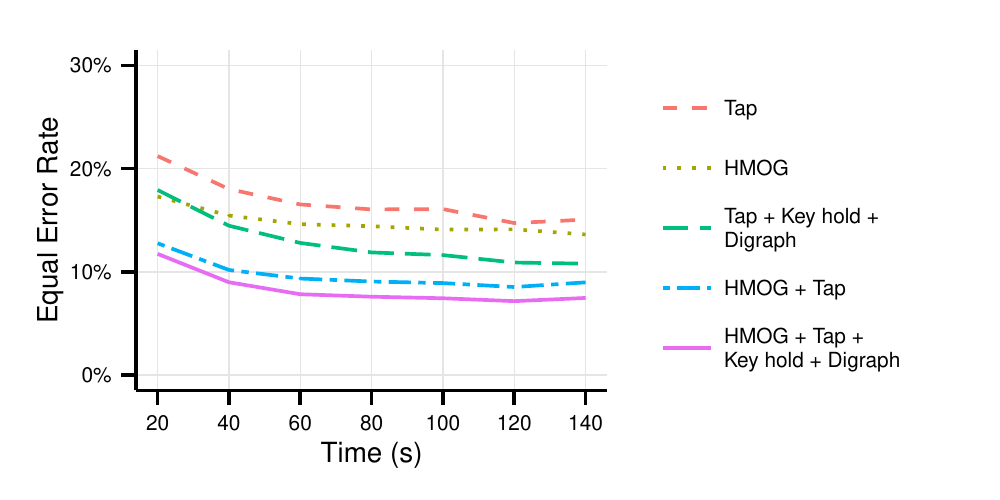}
\label{fig:walkFusionSM}}
\caption{Score-level fusion of combinations of feature types with SM verifier. $X$ axes show authentication time in seconds.}
\end{figure}

\paragraph{Fusion of HMOG, Tap, and Keystroke Features} We used SM verifier and performed score-level fusion with the following feature combinations: \{HMOG, tap, keystroke dynamics\};  \{tap, keystroke dynamics\}; and \{tap, HMOG\}. Detailed fusion results for sitting  and walking conditions are presented in figures \ref{fig:sitFusionSM} and~\ref{fig:walkFusionSM}, respectively. The lowest EERs achieved with fusion are summarized in Table~\ref{tab:fusionperformance},v and the corresponding DET curves for fusion on 60- and 120-second scan lengths are shown in Figure~\ref{fig:DETS}. %

Our results show that: (1) for both walking and sitting conditions, score-level fusion of all signals led to the lowest EER; and (2) fusing HMOG with tap features led to a decrease in EERs and either outperformed (in the case of walking and shorter scans in sitting) or was comparable (in the case of longer scans in sitting) to fusion of tap and keystroke dynamics (see figures~\ref{fig:sitFusionSM}) and \ref{fig:walkFusionSM}). Both (1) and (2) indicate that HMOG provides additional distinctiveness to that of tap and keystroke dynamics, especially in walking condition.

\begin{figure*}[t]
 \centering
\subfigure[Sitting (60-second scans)]
{
\includegraphics[width=0.4\linewidth]{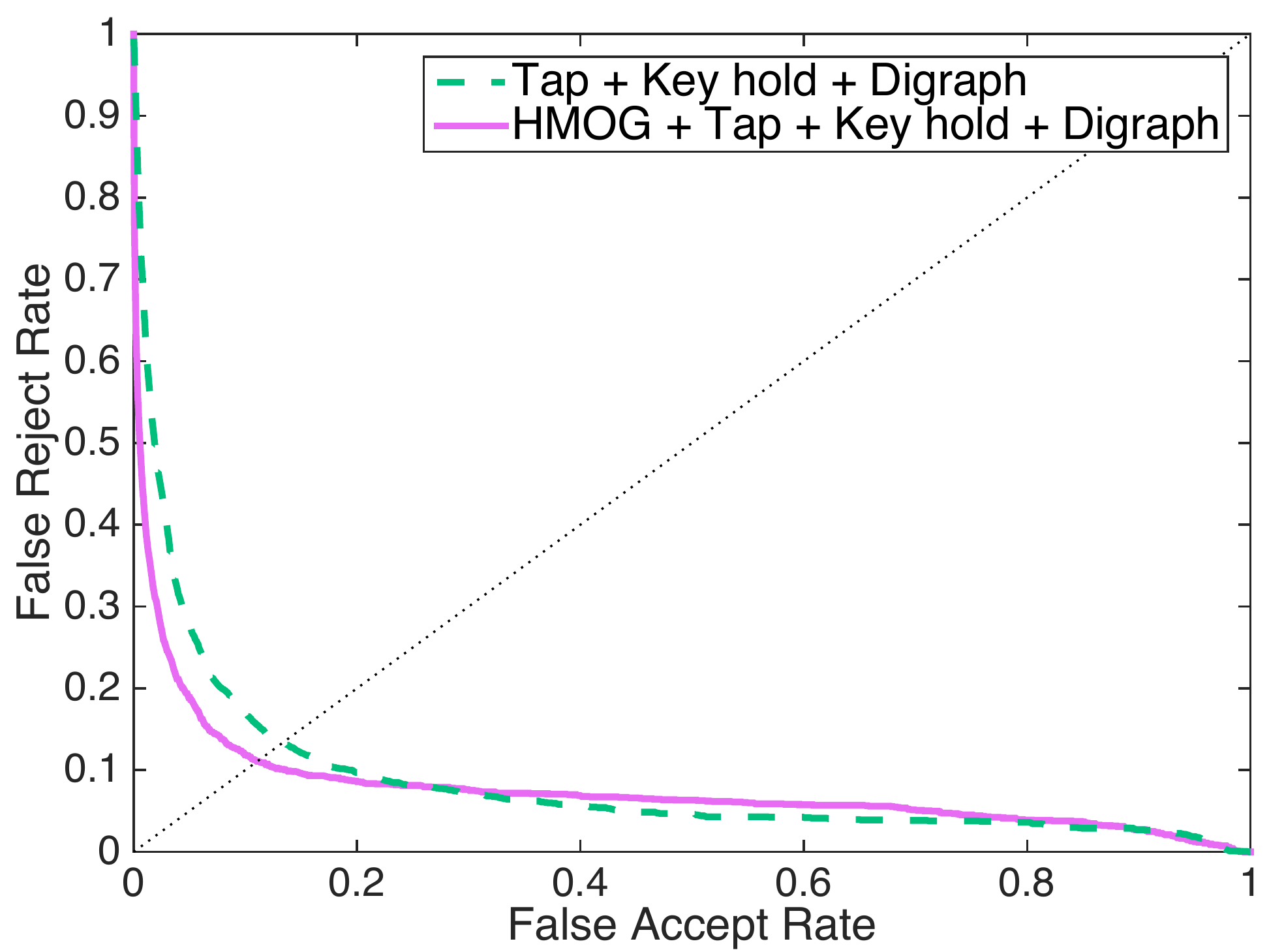}
\label{fig:DETsitFusionSM60}}
\subfigure[Walking (60-second scans)]
{
\includegraphics[width=0.4\linewidth]{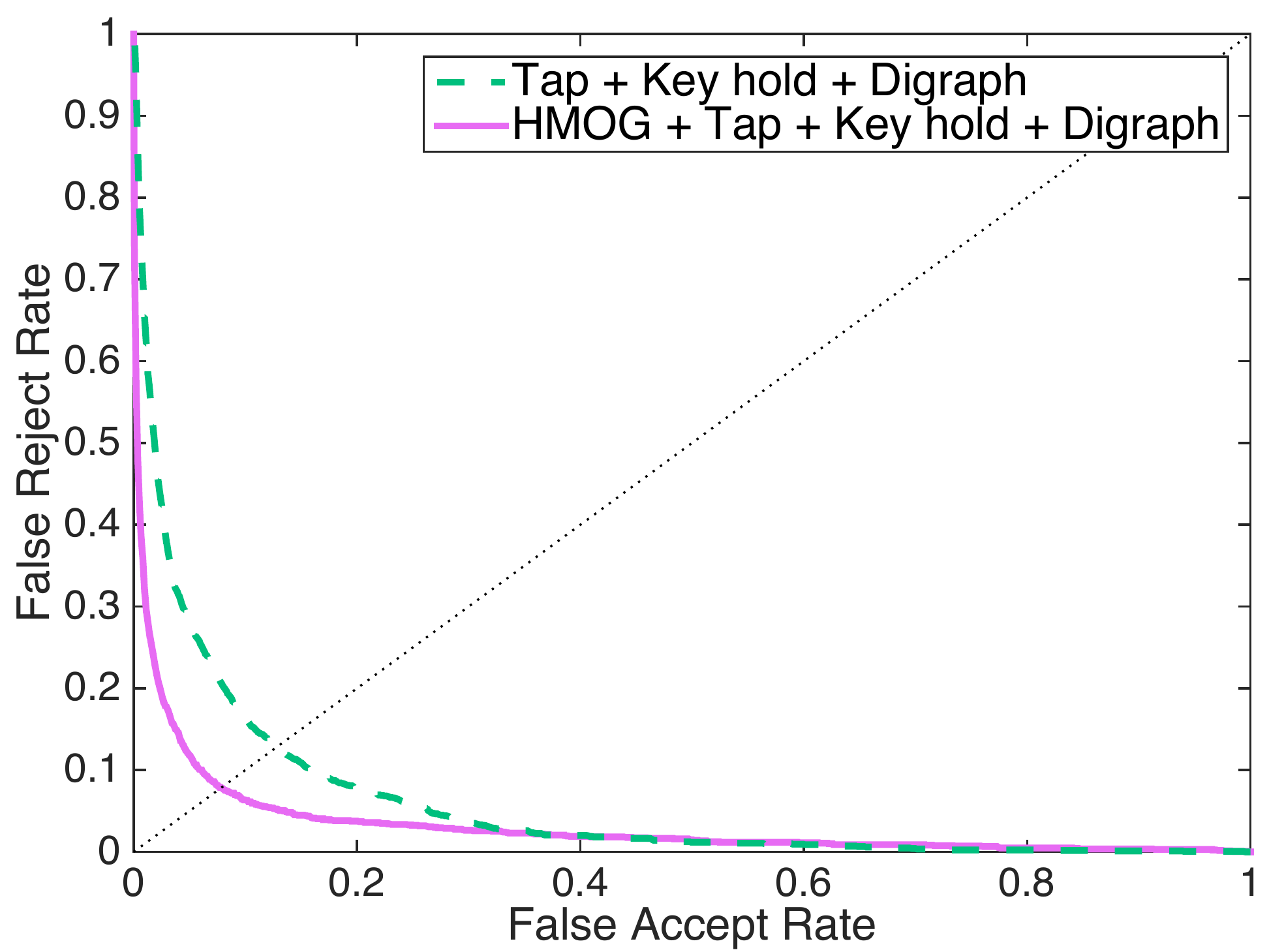}
\label{fig:DETwalkFusionSM60}}
\subfigure[Sitting (120-second scans)]
{
\includegraphics[width=0.4\linewidth]{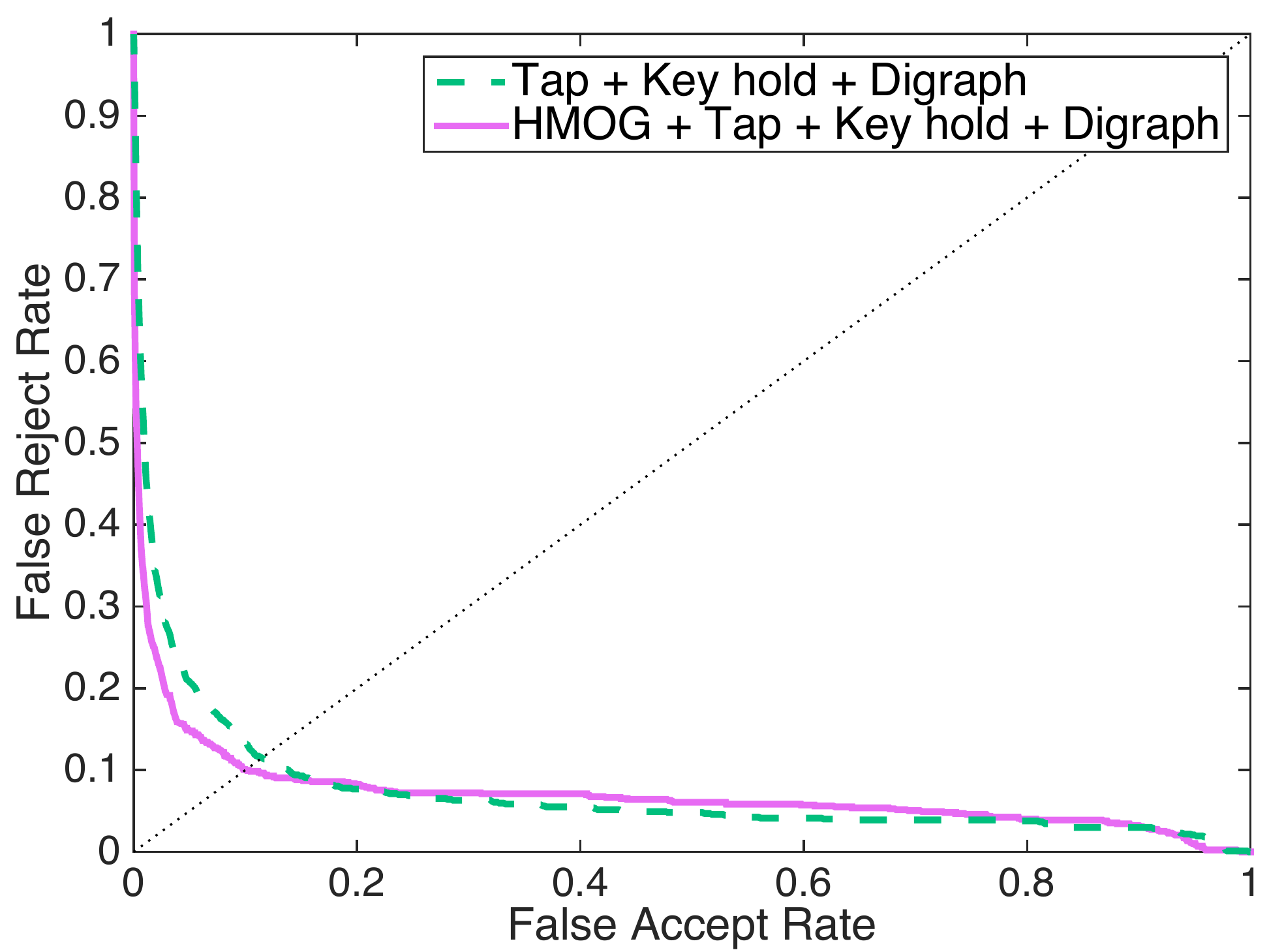}
\label{fig:DETsitFusionSM120}}
\subfigure[Walking (120-second scans)]
{
\includegraphics[width=0.4\linewidth]{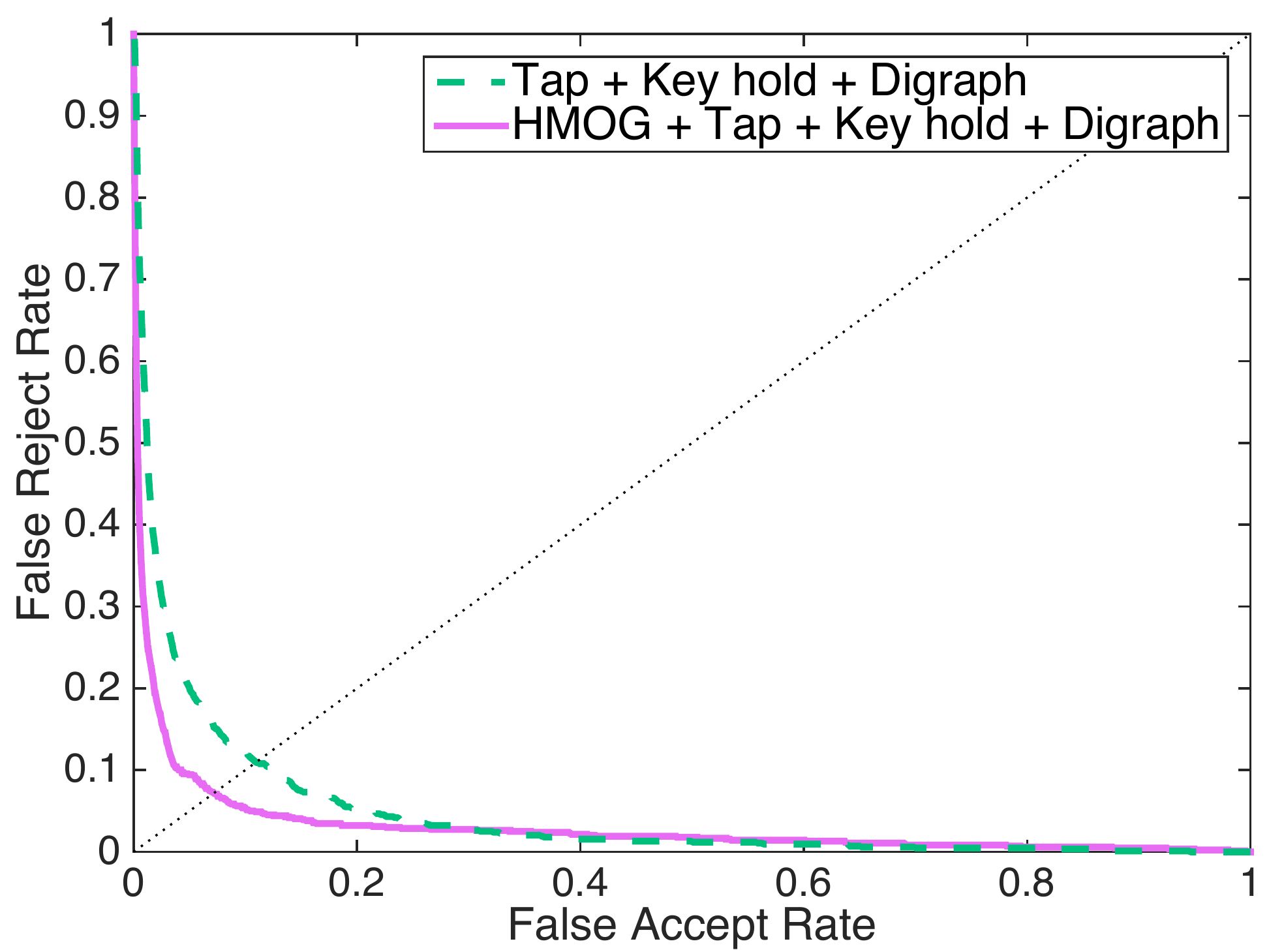}
\label{fig:DETwalkFusionSM120}}
\caption{DET curves for fusion of all feature types including and excluding HMOG. The scan lengths are 60- and 120-seconds.}
\label{fig:DETS}
\end{figure*}

\begin{figure*}[t]
\includegraphics[width=1\linewidth]{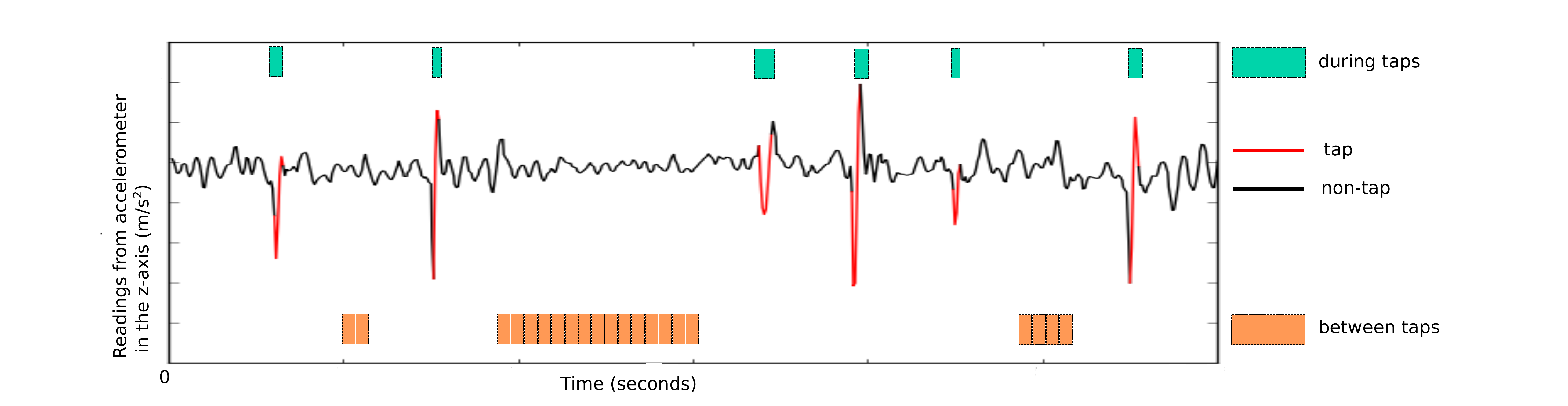}
\caption[]{HMOG features extracted {\em during} and {\em between} taps. The figure shows a sample of readings from the z-axis of accelerometer in sitting condition.}
\label{fig:betweenExplained}
\end{figure*}

\subsection{Why HMOG Features Perform Better During Walking}

\begin{figure}[t]
\includegraphics[width=1.1\linewidth]{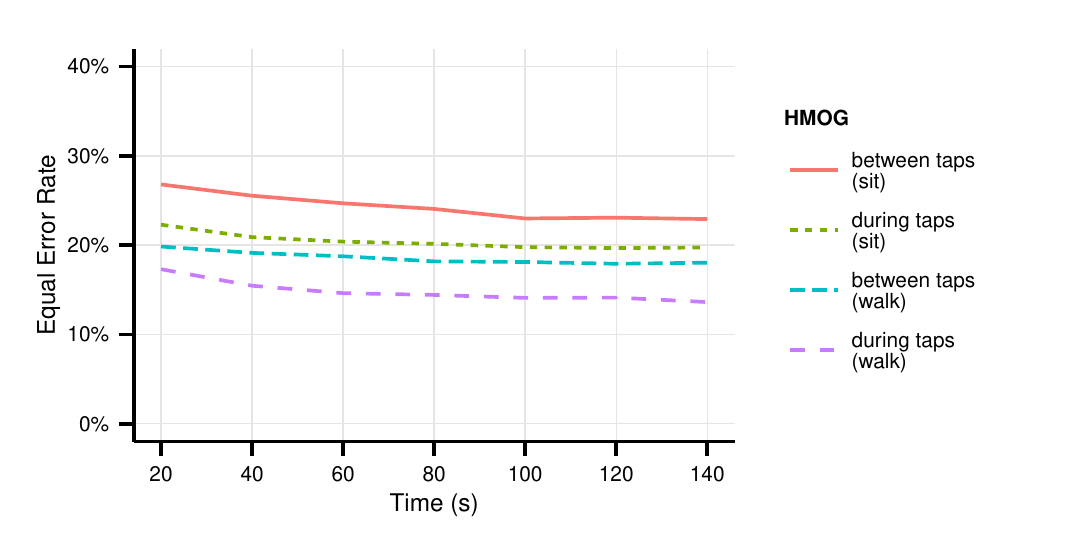}
\caption[]{Performance of HMOG features extracted {\em during}  and {\em between} taps. $X$-axis shows authentication time in seconds.}
\label{fig:betweenTapsSM}
\end{figure}

We investigated why HMOG features performed better during walking. Specifically, we investigated whether the high authentication accuracies of HMOG features during walking were due to hand movements caused by taps, or due to movements caused by walking, or a combination of both. %

\paragraph{Experiment setup} We extracted 64 HMOG features from two segments of an accelerometer/gyroscope signal: (1) {\em during tap}, as discussed in previous sections; and (2) {\em between taps}, in which HMOG features were extracted when the user was \textit{not} tapping the screen (see Figure~\ref{fig:betweenExplained}). In (2), the signal between taps was segmented into non-overlapping blocks of 91~ms; one HMOG feature vector was extracted from each block. 
We selected 91ms as the block size because it was the median duration of a tap in our training data. This ensured that the number of sensor readings used to extract a HMOG feature vector {\em between} and {\em during} tap remained same. %

HMOG features extracted {\em during} taps use sensor readings from 100~ms before and 200~ms after a tap event (see Section \ref{sectionFeatures}). We extracted HMOG features {\em between} taps starting 300~ms after a tap until 300~ms before the next tap, to avoid any overlap between {\em during} and {\em between} HMOG features.%
The average number of the training vectors per user for HMOG {\em during} taps was 1122 for sitting, and 1186 for walking. For {\em between} taps, it was 7692 for sitting and 7462 for walking. %
The average number of testing vectors per user for HMOG features {\em during} taps was 897 for sitting and 972 for walking. For  {\em between} taps, it was 5885 for sitting and 5768 for walking. Verification experiments were performed using SM.

\paragraph{Performance of HMOG Features Extracted During vs. Between Taps}  
We compared HMOG features extracted during taps with the same features extracted between taps for sitting and walking conditions. For sitting, HMOG features extracted {\em during} taps performed consistently better  than those extracted {\em between} taps (see EERs in Figure~\ref{fig:betweenTapsSM}). This indicates that HMOG features were able to capture distinctive hand micro-movement patterns when the users tapped on the phone. %
Similarly, for walking, HMOG features extracted {\em during} taps performed better than those extracted {\em between} taps (see EERs in Figure~\ref{fig:betweenTapsSM}). This again indicates that HMOG features capture user's distinctive hand micro-movement patterns when the user is tapping, regardless of the motion condition.

\paragraph{Impact of Walking on HMOG Features Extracted Between Taps} HMOG features extracted  {\em between} taps during walking outperformed the same when extracted during sitting (see {\em between} tap EERs for sitting and walking in Figure~\ref{fig:betweenTapsSM}). This indicates that HMOG features capture distinctive movements induced by walking, even in the {\em absence} of tap activity.%

\bigskip
Supported by the above results, the high authentication accuracies achieved by HMOG features during walking can be jointly attributed to: (a) the distinctiveness in hand movements caused by tap activity and (b) the distinctiveness in movements caused by walking. %

\section{Biometric Key Generation \\from HMOG features}
\label{sec:bkg}

\begin{table*}[t]
\caption{List of symbols used in this section.}
\begin{tabular}{|l|l||l|l|}
\hline
$C$ 							& Error-correcting code (vector space over $\Z_p$). & $\delta = (x-c)$					& Discretize biometric sample $x$ masked by $c$. \\

$n$ 							& Number of biometric features, as well as length of $C$. & $\gamma$					& Fuzzy commitment \\
$p$ 							& Alphabet size of $C$. ($p$ is a prime number larger $n$.)  & ${\sf d\_range}_{{F}_i}$			& Upper bound of the range of ${\sf DS}_{{F}_i}$ \\
$l$							& Dimension of $C$ (size of the basis of $C$).  & ${\sf min}_{{F}_i}, {\sf max}_{{F}_i}$ & Minimum and maximum values of $F_i$	 \\ 
$w_L(\cdot), d_L(\cdot,\cdot)$			& Lee weight and Lee distance functions. & $\sigma_i$				& Standard deviation of $F_i$ \\
$z$							& System-wide public constant or user PIN/password. & $x_i$  		& Instance of $F_i$ \\ 
$c = (c_1, \ldots, c_n)$			& Codeword from $C$. & ${F}_i$		& Biometric feature $i$\\
$x = (x_1, \ldots, x_n)$			& Discretized biometric sample. & & \\

\hline
\end{tabular}
\label{tbl:bkg-symbols}
\end{table*}

In this section, we evaluate the performance of HMOG features for biometric key 
generation (BKG). For this purpose, we introduce our BKG construction, which 
extends and generalizes the fuzzy commitment scheme of Juels et 
al.~\cite{fuzzy_commitment}. While the technique in~\cite{fuzzy_commitment} operates on features represented using a single bit, our BKG construction represents features as symbols of an alphabet of arbitrary prime size $p$. Our BKG relies on Reed-Solomon error correcting codes~\cite{rot06} in Lee metric~\cite{lee58}. Notation used in this section is presented in Table~\ref{tbl:bkg-symbols}.

\paragraph{Preliminaries}
BKG uses  biometric information to prevent unauthorized access to cryptographic keys; these keys can then be used, e.g., to encrypt/decrypt sensitive information. The process of protecting a key is referred to as {\em 
committing}, and the outcome of this process is a {\em commitment}. Given a 
commitment, the cryptographic key is reconstructed by {\em decommiting} (or {\em 
opening}) it, using information from a biometric signal. 
Informally, a BKG construction is secure if a key committed using a biometric signal $s$ can be opened only using a signal $s'\approx s$, and both $s$ and $s'$ are from the same user. 

BKG techniques use error-correcting codes to address natural variations among 
different biometric samples from the same users.  An error-correcting code is 
defined as a set $C$ of codewords. Typically, there are two functions associated 
with a code: ${\sf encode(\cdot)}$ and ${\sf decode(\cdot)}$. The former maps a 
message to a codeword; the latter---a possibly perturbed codeword to the original 
codeword. 
The ${\sf decode(\cdot)}$ function is designed to maximize the probability of 
correct decoding. 

\subsection{Our Construction}
\paragraph{Scaling and Discretization}
BKG techniques work on discrete values, instead of real values. Therefore, the 
user first performs scaling and discretization of the feature vector representing 
her biometric.
Each feature $F_i$ is assigned a discretization range $[0, {\sf d\_range}_{F_i}]$, where ${\sf d\_range}_{F_i} \in \{(p-1)/2,\ldots, p-1\}$ is negatively correlated to the standard deviation $\sigma_i$ of $F_i$ (i.e., if $\sigma_i < \sigma_j$, then ${\sf d\_range}_{F_i} > {\sf d\_range}_{F_j}$).

Let $x_{i}$ be an instance of ${F_i}$, and ${\sf min}_{F_i}$ and ${\sf max}_{F_i}$ be respectively the typical minimum and maximum value of ${F_i}$. Discretization and scaling are performed as:
\[\small
{\sf DS}_{{{F_i}}}(x_i) = 
\resizebox{0.8\linewidth}{!}{$
\begin{cases}
0 & x_i<{\sf min}_{F_i} \\ \\
\left\lfloor {\sf d\_range}_{F_i}\cdot\left(\frac{ x_i-{\sf min}_{F_i}}{{\sf max}_{F_i}-{\sf min}_{F_i}}\right)\right\rfloor & {\sf min}_{F_i} \leq x_i \leq {\sf max}_{F_i} \\ \\
{\sf d\_range}_{F_i} & x_i>{\sf max}_{F_i} \\
\end{cases}$}
\]

\paragraph{Committing a Key}
To commit a cryptographic key using $n$ biometric features, the user 
selects a random codeword $c$ of length $n$ from $C\subset(\Z_p)^n$. The key is 
computed as $k = \PRF_c(z|0)$, where $\PRF$ is a pseudorandom function family, 
$z$ is a system-wide public constant and ``$|$'' denotes string concatenation. (BKG can be augmented with a second authentication factor by setting $z$ to a user-provided password.) $c$ is then committed using the user's biometrics as discussed next. 
Let $x = (x_1, \ldots, x_n)$ be a scaled and discretized feature vector. The user 
computes $\delta = (x-c) = (x_1-c_1, \ldots, x_n-c_n)$ and publishes commitment $
\gamma = (\PRF_c(z|1), \delta)$.

The user computes $k$ from $\gamma$ and her biometric signals (and possibly a 
password $z$) as follows. She extracts biometric features from the signal, and 
encodes them as $y=(y_1,\ldots,y_n)$. Then, she computes $c' = {\sf decode}(y-
\delta)$. If $\PRF_{c'}(z|1) = \PRF_c(z|1)$, then $k=\PRF_{c'}(z|0)$ with overwhelming probability. 

Asymptotic complexity of BKG key retrieval is dominated by one instance of Euclidean algorithm and one matrix-vector multiplication, both in $O(n^2)$ finite field operations in a field of size $p \geq n$. Security of our construction is analyzed in Appendix~\ref{sec:security}.  %

\paragraph{Using Lee-metric Decoding for BKG}
Distance between feature vectors is defined using the Lee distance~\cite{lee58}---a discrete approximation of SM:

\begin{definition}[Lee weight]
Let $p$ be an odd prime. The \emph{Lee weight} of element $x \in \Z_p$ is defined as $w_L(x) = \min|x'|$, for $x' \equiv x \bmod{p}$. %
The \emph{Lee weight} of vector $x = (x_1,\ldots, x_n) \in (\Z_p)^n$ is defined as the sum of Lee weights of its elements, i.e., $w_L(x) = \sum_{i=1}^n{w_L(x_i)}$.
\end{definition}

\begin{definition}[Lee distance]
The \emph{Lee distance} of vectors $x, y \in \Z_p$ is the Lee weight of their difference, i.e., $d_L(x, y) = w_L(x - y)$. 
\end{definition}
In $\Z_2$, the Lee weight coincides with Hamming weight.

We used normalized generalized Reed-Solomon codes from \cite{rot06}, presented next, to implement the ${\sf encode}(\cdot)$ and ${\sf decode}(\cdot)$ functions.

\begin{definition}
Let $l \leq n$ and $n\leq p$. A linear $[n, l]$-code over $\Z_p$ is a $l$-dimensional vector subspace of $(\Z_p)^n$. 
 A normalized Reed-Solomon $[n, l]$-code over $\Z_p$ is a linear $[n, l]$-code over $\Z_p$ with parity-check matrix: 
\[ \resizebox{0.55\linewidth}{!}{$
 H = \left( 
 \begin{array}{cccc}
1 & 1   & \ldots & 1 \\
1 & 2   & \ldots & n \\
1 & 2^2 & \ldots & n^2 \\
\multicolumn{4}{c}{\vdots}   \\
1 & 2^{n-l-1} & \ldots & n^{n-l-1} \\
 \end{array} 
 \right)$}
\]
and generator matrix:
\[ \resizebox{0.55\linewidth}{!}{$
 G = \left( 
 \begin{array}{cccc}
v_1 &     v_2 & \ldots &     v_n \\
v_1 & 2   v_2 & \ldots & n   v_n \\
v_1 & 2^2 v_2 & \ldots & n^2 v_n \\
\multicolumn{4}{c}{\vdots}   \\
v_1 & 2^{l-1} v_2 & \ldots & n^{l-1} v_n \\
 \end{array} 
 \right) $}
\]
\end{definition}

The rows of $G$ form a basis of the nullspace of $H^T$. 

To obtain a random codeword from $C$, we select a $l$-tuple $m = (m_1, \ldots, m_l) \in (\Z_p)^l$ uniformly at random, and encode $m$ as $c = mG$. 
The Lee distance of $C$ is $2(n-l)$, so for any error $e = (e_1, \ldots, e_n)$ with $w_L(e) < n-l$, ${\sf decode}(c+e) = c$.

\subsection{Evaluation of HMOG Features on BKG}

To evaluate HMOG features for BKG, we determined their authentication 
accuracy and security against population attacks. We then compared our results with HMOG features to that of tap, key hold, and swipe features under the same metrics.

\paragraph{Biometric Accuracy}
Figures~\ref{fig:bkg_sitting} and~\ref{fig:bkg_walking} summarize the results of 
our experiments, performed on our 100-user dataset. We evaluated BKG using the 
features that performed best for authentication. We used 17 and 13 HMOG features in sitting and walking conditions, respectively (see Figure~\ref{fig:fisherFeatures}).

We ran experiments on four different feature subsets: (1)~HMOG-only features; 
(2) 11 tap-only features; (3) 12 (sitting) and 8 (walking) key hold-only features; 
and (4) HMOG and 3 best-performing tap features. For both walking and sitting experiments, feature 
subset (4) provided the best results, i.e., 15\% and 20\% EER respectively, for both one-minute and two-minute scan lengths. For sitting 
experiment, key generation was not possible with (3), as the within-user 
variability of the biometrics signals was too high. 

Sedenka et al.\cite{sedenka2014} showed that Linear Discriminant Analysis (LDA)~\cite{lda} 
improves BKG accuracy on desktop keystroke dynamics. However, HMOG, tap, and 
key hold features on a virtual keyboard did not benefit from LDA. Therefore we do 
not report BKG results with LDA for these features.

\paragraph{Security Against Population Attacks}
EER computed via zero-effort attacks provides limited information on the security 
of a BKG scheme, because it does not take into account all the information 
readily available to the adversary. In particular, with BKG the adversary has 
access to: (1) the commitment~$\gamma$; and (2) an approximation of the 
distribution of the user's biometric signals obtained from  {\em population 
data}. 

Access to $\gamma$ allows the adversary to test whether a particular feature 
vector decommits the key. The adversary can perform this test {\em offline}, 
i.e., with no restrictions on the number of attempts performed (within the limits of the available resources). Therefore, the hardness of ``guessing'' a user's feature vector given 
$\gamma$ is an upper bound on the security of a BKG scheme. 

The adversary can use (2) to guess the user's feature vector more efficiently, 
under the assumption that biometric signals from different users are not 
completely independent. To this end, Ballard et al.~\cite{bal08} proposed the notion 
of \emph{guessing distance}. It is defined as the logarithm of the number of 
guesses necessary to open a 
commitment using feature vectors from multiple impostors.

We instantiated guessing distance in our setting as follows. First, we built a 
commitment $\gamma_i=(\PRF_{c_i}(z|1), \delta_i)$ from the feature vector 
of each user $i$. Then, we used the biometric sample from user $j$ to open all $
\gamma_i$ such that $i\neq j$, and ranked users according to how many 
commitments they were able to open.  Finally, for each user $i$ we select users 
$j\neq i$ following to this ranking, and determined how many attempts were 
required to open $\gamma_i$. Guessing distance was computed as the binary 
logarithm of this value.
There might be users $i$ for which no feature vector from other users could open 
$\gamma_i$. We refer to the commitments of these users as {\em non-guessed}.

Table~\ref{tbl:bkg_guessing} summarizes the results of our experiments for one 
minute scans. The lowest EER was achieved by combining HMOG features 
with the tap features selected by mRMR. Overall, HMOG outperformed tap 
features for biometric key generation. Nevertheless, our results show that most 
commitments can be guessed using population data. 

\paragraph{Comparison of HMOG and Swipe Features}
Because there is no previous work on BKG using touch-, accelerometer-, or gyroscope-based features, we compared BKG on HMOG with BKG on touch features extracted from swipes ({\em swipe features} hereafter). For this purpose, we computed swipe features from the datasets of Serwadda et al.~\cite{serwadda2013}. As in \cite{serwadda2013}, we used the whole first session for computing commitments and ten swipes from the second session to perform impostor/genuine {\em open} attempts. 
Results are reported in Table \ref{tbl:bkg_guessing}. 

We also performed experiments on the touch dataset of Frank et al.~\cite{frank2013}. However, a large majority of the users could not reliably decommit their own keys. This was due to the large variance between the vectors used to build the commitment, and the one used to open it. Therefore, we did not include these results in this paper.

\begin{table*}[ht!]
\centering
\caption{Security of BKG based on HMOG, tap, key hold, and swipe features. %
}
\label{tbl:bkg_guessing}

\begin{tabular}{llllll}
Dataset                              & Features                                                                    & EER    &  \parbox{1.5cm}{Average guessing distance} & \parbox{2cm}{Non-guessed commitments} & $\log_2(|C|)$ \\ \hline
\multirow{3}{*}{Our dataset-sitting} & HMOG                                                                           & 23.4\% & 2.8         & 2\%                    & 19                  \\
                                     & HMOG with best 3 tap                                                         & 20.1\% & 2.7         & 1\%                    & 18                  \\
                                     & Tap                                                                          & 25.7\% & 1.6         & 2\%                     & 25                  \\ \hline
\multirow{4}{*}{Our dataset-walking} & HMOG                                                                           & 17.4\% & 2.9         & 2\%                    & 27                  \\
                                     & HMOG with best 3 tap                                                         & 15.1\% &3.2        & 5\%                    & 33                  \\
                                     & Tap                                                                          & 28.4\% & 1.9         & 0\%                     & 30                  \\
                                     & Key hold                                                                        & 28.9\% & 1.9         & 0\%                     & 10                   \\ \hline
Serwadda et al.~\cite{serwadda2013}                     & \parbox{3cm}{portrait orientation, vertical swipes, with LDA}                  & 34.2\% & 3.3         & 0\%                     & 39                  \\ %
\end{tabular}
\end{table*}

\begin{figure}[t!]
 \centering
\subfigure[EER for Sitting]
{
\includegraphics[width=1.1\linewidth]{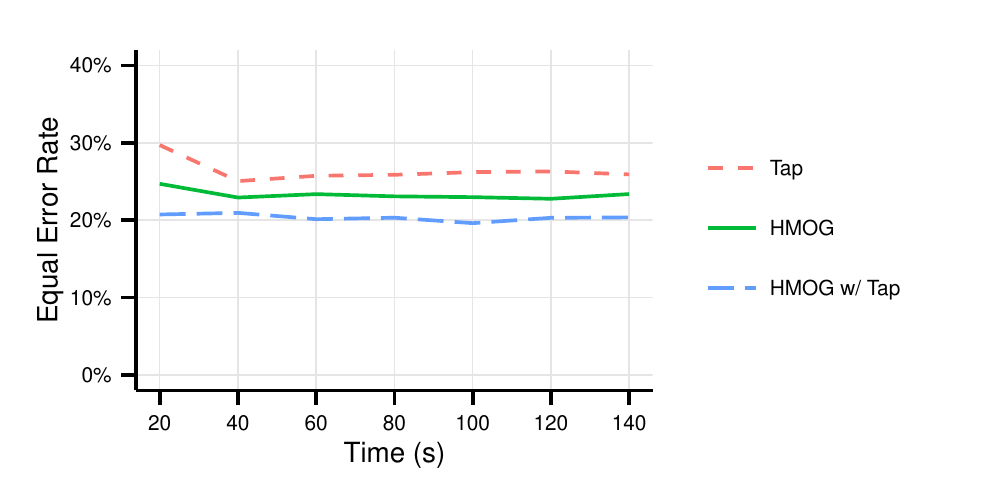}
\label{fig:bkg_sitting}}
\subfigure[EER for Walking]
{
\includegraphics[width=1.1\linewidth]{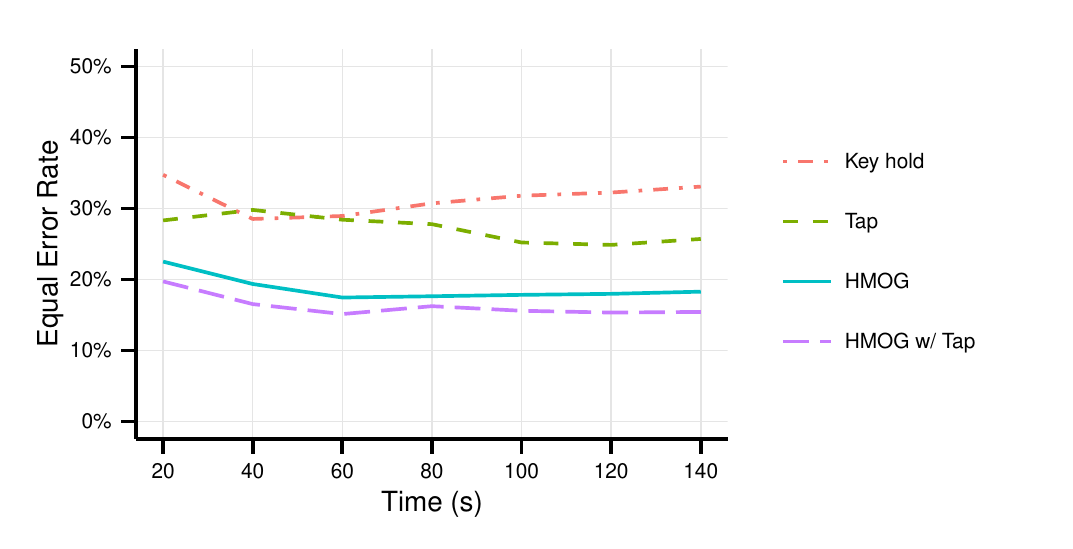}
\label{fig:bkg_walking}}
\caption{EERs for BKG during sitting and walking. $X$-axis shows the authentication time in seconds.} %
\end{figure}

\section{Energy Consumption of HMOG Features}
\label{sec:power}

We measured the energy consumption of two basic modules involved in extracting HMOG features:  (1) sensors (i.e., accelerometer and gyroscope); and (2) feature computation (i.e., calculation of a HMOG feature from raw sensor readings). Our main finding from energy consumption analysis is that decreasing sensor sampling rate (from 100Hz to  as low as 16Hz) considerably reduced the energy overhead without impacting the authentication performance of HMOG features. %

\subsection{Experiment Setup and Design}
\label{sec:powerms}

We developed an Android application that collects and processes sensor data at 
different sampling rates. Our application allows us to selectively enable sensors and HMOG features. We extracted the best-performing 17 features for sitting (i.e., top-ranked 17 features in Figure~\ref{fig:fisherFeaturesSit} that were selected by 10-CV) and 13 HMOG features for walking (i.e., top-ranked 13 features in Figure~\ref{fig:fisherFeaturesWalk}, selected by 10-CV). The union of these two feature sets resulted in 18 HMOG features. Because none of these feature were extracted from magnetometer, we did not measure its energy consumption. 

 Experiments were performed using a Samsung Galaxy S4 smartphone running Android 4.4. To obtain consistent and repeatable results, we terminated all other applications 
and all Google services on the smartphone. Additionally, we switched off WiFi, 
Bluetooth, and cellular radios. The screen was turned on during the experiments. 
Automatic brightness adjustments were disabled, and brightness was set to the 
lowest level. We used the Monsoon Power Monitor ~\cite{monsoonPowerMonitor} to 
measure the phone's energy consumption.

We performed the energy consumption experiments as follows. First, we measured baseline energy 
consumption by running our application with all sensors and features disabled. 
Then, we enabled accelerometer and gyroscope, and evaluated the corresponding energy consumption.
Our application is designed to sample sensors at all supported frequencies. In the case of Galaxy S4, the available sampling rates are: 5Hz, 16Hz, 50Hz, and 100Hz. 
We used authentication scan lengths of 60 and 120 seconds. Our results report the average and standard deviation of ten experiments in each setting. 
Finally, we quantified the energy overhead of computing 18 HMOG features from sensor readings acquired during data collection.

\paragraph{Calculation of EERs at Lower Sampling Rates} We originally collected our data at 100Hz sampling rate. In order to obtain the EERs for lower sampling rates, we used downsampling. For example, to simulate 16Hz sampling rate, we choose every sixth sensor reading from the original sensor data. Then, using the downsampled data, we performed HMOG-based authentication with SM verifier for 60- and 120-second scans using the same evaluation process as in Section  \ref{sectionExperiments}. 

\subsection{Energy Consumption of HMOG Authentication}
\label{sec:pr}

\paragraph{EERs vs.~Energy Consumption} 
Figure~\ref{fig:eer_diff_rate} shows that  EERs for 16Hz 
sampling rate are comparable to those of 50Hz and 100Hz for both sitting and walking, while the EERs for 5Hz are 
considerably worse than 16Hz, 50Hz, and 100Hz. 

On the other hand, Table \ref{tbl:power_result} shows that energy overhead over the baseline is low (between 
6.2\% and 7.9\%) for 5Hz and
16Hz sampling rates and, in comparison, high (between 12.8\% and 20.5\%) for 50Hz 
and 100Hz. 
Thus, during the active authentication with HMOG, we can choose 16Hz instead of 
100Hz as the sensor sampling rate, which would lower the energy overhead of 
sensor data collection by about 60\% without sacrificing EER. 

\paragraph{Energy Consumption for Feature Computation} 
The energy overhead for computing 18 
features is very low compared to energy overhead of sensor data collection. The  
energy consumption for computing all 18 HMOG features is 0.08 joules, 
which corresponds to  0.19\% overhead for the 60-second and  0.1\% for  
120-second scans. This low overhead can be attributed to the fact that HMOG features are time-domain 
features. (As suggested by previous 
research~\cite{Krause:2005:TOP:1104998.1105279,Yan:2012:ECA:2357489.2358011}, 
computing time-domain features consumes less energy than computing frequency-domain features.)
Further, because HMOG feature computation involves simple arithmetic calculations, 
they can be processed very quickly by the smartphone's CPU (on average 37ms per feature,  in our case). %

\begin{table}[t]

\centering 
\caption{Energy consumption measurement results.}
\label{tbl:power_result}
\resizebox{1\columnwidth}{!}{
\begin{tabular}{c|llllll|}
{} & \multicolumn{6}{c}{60 Seconds Scan Length} \\
{} & {} & Baseline & 5Hz & 16Hz & 50Hz & 100Hz \\\hline
Energy & Mean & 42.7 & 45.5 & 46.1 & 48.4 & 51.5 \\
Consumption (J) & StdDev & 0.04 & 0.13 & 0.15 & 0.25 & 0.29 \\\hline
{Overhead to Baseline} & {} & N/A & 6.6\% & 7.9\% & 13.3\% & 20.5\% \\
\multicolumn{7}{c}{ }  \\ 
{} &  \multicolumn{6}{c|}{120 Seconds Scan Length} \\ 
{} & {} & Baseline & 5Hz & 16Hz & 50Hz & 100Hz \\\hline
Energy & Mean & 85.6 & 90.9 & 92.1 & 96.5 & 102.9 \\
Consumption (J) & StdDev & 0.17 & 0.21 & 0.23 & 0.30 & 0.36\\\hline
{Overhead to Baseline} & {} & N/A & 6.2\% & 7.6\% & 12.8\% & 20.1\%
\end{tabular}}
\end{table}

\begin{figure}[!htb]
\includegraphics[width=1.1\linewidth]{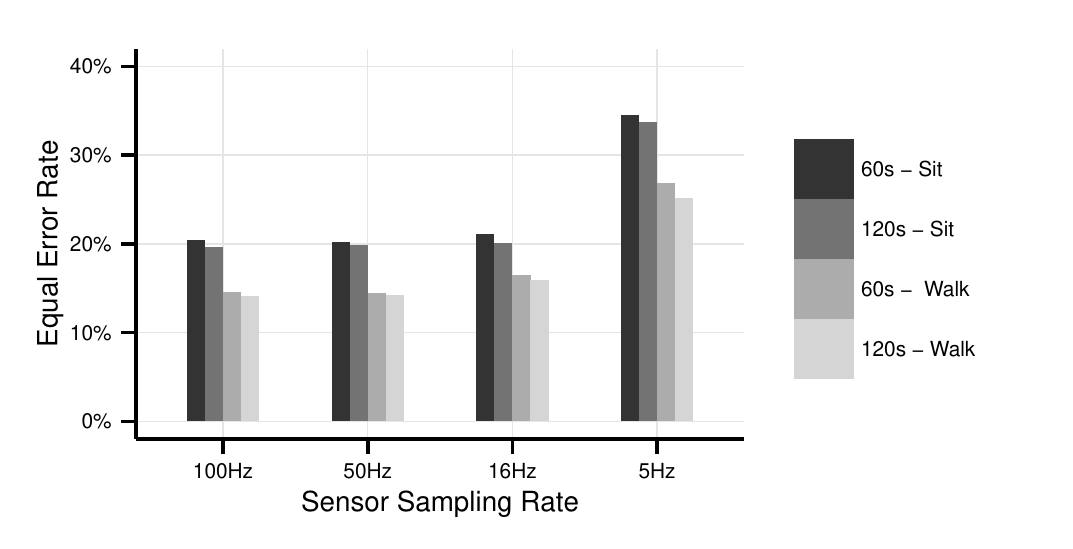}
\caption[]{Performance of HMOG features with different sensor sampling rates using SM verifier.}
\label{fig:eer_diff_rate}
\end{figure}

\begin{table*}[t]
\caption{Comparison of our authentication experiments with related work on smartphone tap/typing authentication.}
\centering
\resizebox{1\textwidth}{!}{

\begin{tabular}{| c | c | c | c | c | c | c | c | c | c | c | c | c |}
\hline
Work  & Condition        & \begin{tabular}{@{}l@{}}  Free text \end{tabular} & {Motion-based features}  &  \begin{tabular}{@{}l@{}}  Tap features \end{tabular} &  \begin{tabular}{@{}l@{}}  Keystroke \\ features \end{tabular} & \begin{tabular}{@{}l@{}}  Authentication \\ vector \end{tabular} &  \begin{tabular}{@{}l@{}} \# of \\ owners \end{tabular} & \begin{tabular}{@{}l@{}}Avg. \# of \\ impostors \\ per owner\end{tabular} & Verifier                                        & \begin{tabular}{@{}l@{}}  Training \\data source  \end{tabular}              & Best FAR                           & Best FRR              \\ \hline
Trojahn et al.~\cite{trojahn2012} & sit          & \xmark       & \xmark        &  \begin{tabular}{@{}l@{}}  pressure, \\ contact size \end{tabular} &  digraph & \begin{tabular}{@{}l@{}}  avg. of 7 samples \end{tabular}          & 35                & 34                                    & ANN                                             & {unknown}                         & 9.53\%                               & 5.88\%                  \\ \hline
Li et al.~\cite{li2013unobservable} & \multicolumn{2}{c|}{regular smartphone usage}       & \xmark        &  \parbox{2cm}{pressure, touch area, duration} &  \xmark & avg. of 2--20 gestures          & 28                & up to 47                                    & \parbox{1.5cm}{SVM (Gaussian kernel)}                                        & \begin{tabular}{@{}l@{}}owner \& \\impostor\end{tabular}                         & \multicolumn{2}{c|}{not reported for taps}                \\ \hline
Zheng et al.\cite{zheng2012}   & sit           & \xmark        & \begin{tabular}{@{}l@{}} min, max and mean of accele-\\ration and angular velocity at \\press, release of each PIN digit    \end{tabular}      &  \parbox{2cm}{pressure, contact size (both at \linebreak press and release)} &  \begin{tabular}{@{}l@{}}  key hold, \\ key interval \end{tabular} & each tap             & 80                & 79                                    & \begin{tabular}{@{}l@{}}  dissimilarity \\ score \end{tabular}                            & owner                       & \multicolumn{2}{c|}{EER = 3.65\%}             \\ \hline
Feng et al.~\cite{feng2013}    & sit        & \xmark         & \xmark       &   pressure   &  \begin{tabular}{@{}l@{}}  key hold,  \\ key interval \end{tabular} & \begin{tabular}{@{}l@{}}  5-60  char. window  \end{tabular}          & 40                & 39                                    & \begin{tabular}{@{}l@{}}  decision tree, \\ Bayesian \\ networks, \\random forest \end{tabular} & \begin{tabular}{@{}l@{}}owner \& \\impostor\end{tabular}                        & \multicolumn{2}{c|}{EER = 1\%}\\ \hline
Gascon et al.~\cite{gascon2014}   & sit           & \xmark        & \begin{tabular}{@{}l@{}} accelerometer, gyroscope, and \\ orientation features extracted\\ during  typing burst\end{tabular}       &  \xmark & \xmark & \begin{tabular}{@{}l@{}}  features extracted \\ from time window \end{tabular}            & 12                & 303                                   & linear SVM                                      & \begin{tabular}{@{}l@{}}owner \& \\impostor\end{tabular}                           & \begin{tabular}{@{}l@{}}  \parbox{1.95cm}{1\% (4 genuine  users)} \end{tabular}         & \begin{tabular}{@{}l@{}}  \parbox{1.4cm}{8\%  (4 genuine  users)} \end{tabular} \\ \hline
Bo et al.~\cite{bo2013}      & \begin{tabular}{@{}l@{}}  sit, walk \end{tabular}  & \cmark       & \begin{tabular}{@{}l@{}}  mean magnitude of acceleration\\ and angular velocity during tap  \end{tabular}        &  \begin{tabular}{@{}l@{}}  \parbox{1.9cm}{coordinate, pressure, duration} \end{tabular} &  \xmark & \begin{tabular}{@{}l@{}}  each gesture, \\ judgement after  \\ 1-13 gestures \end{tabular}            & 10                & 50                                    & SVM                                             & \begin{tabular}{@{}l@{}} owner only, \\ and owner \& \\ impostor \end{tabular} & \begin{tabular}{@{}l@{}}  0\% when trained \\with owner \& \\ impostor data, 24.99\% \\ with owner data   \end{tabular}     & \begin{tabular}{@{}l@{}}  0\% when \\ trained with \\owner data \end{tabular} \\ \hline
This work & \begin{tabular}{@{}l@{}}  sit, walk \end{tabular}  & \cmark      & \begin{tabular}{@{}l@{}}  60 resistance and 36 stability\\ features extracted from tap, \\using accelerometer, gyroscope\\ and magnetometer \end{tabular}   &  \parbox{1.95cm}{contact size (9 f.),  duration, velocity between two taps} &  \begin{tabular}{@{}l@{}}  key hold, \\ digraph \end{tabular} & \begin{tabular}{@{}l@{}} taps averaged \\ in time window \\ (20, 40, 60, 80, \\ 100, 120, 140 s.) \end{tabular}             & 90                & \begin{tabular}{@{}l@{}}99 sit \\ 93 walk \end{tabular}               &  \begin{tabular}{@{}l@{}} SM, SE, \\ 1-class SVM  \end{tabular}                           & owner                      & \multicolumn{2}{c|}{EER = 7.16\%}                  \\ \hline
\end{tabular}

}
\label{tab:relatedResearch}
\end{table*}

\section{Related Work} \label{relatedResearch}

\paragraph{Evolution of Continuous Authentication in Desktops and Mobile Phones}
The need to periodically authenticate the user after login, combined with the fact that behavioral biometric traits can be collected without interrupting the user, led to promising research in the area of continuous authentication. Early work in the field used keystroke dynamics~\cite{Gunetti:2005:KAF:1085126.1085129, Dowland:2002:KAM:647185.719834, Monrose:1997:AVK:266420.266434, dowland2002keystroke} to authenticate desktop users. Later studies on desktop users demonstrated the feasibility of using a variety of behavioral traits, including mouse dynamics~\cite{shen2013user}, soft-biometrics~\cite{5570993}, hand movement~\cite{6879297}, keyboard acoustics~\cite{6966780}, screen fingerprints~\cite{Fathy2014122},  language use~\cite{stolerman11active, pokhriyal2014use} and cognition during text production~\cite{locklear2014, monaco2012, stewart2011, monaco2013}.

Early studies in continuous authentication of mobile phone users focused on keystroke dynamics (see~\cite{clarke2007,clarke2007109,buchoux2008,maiorana2011,campisi2009}),  because these devices had a hardware keyboard to interface with the user. However, as mobile phones evolved into ``smartphones'',  research in this area has been reshaped to leverage the multitude of available sensors  on these devices (e.g., touchscreen, accelerometer, gyroscope, magnetometer, camera, and GPS). Two behavioral traits have been predominantly explored in the smartphone domain, (1) gait (see, e.g.,~\cite{frank2013,serwadda2013}), and (2) touchscreen interaction (see, e.g.,~\cite{derawi2010,vildjiounaite2006}). More recently, research has focused on  leveraging multi-modal behaviors (e.g.,~\cite{bo2013,Weidong2011}).

\paragraph{Continuous Authentication Using Taps}
Because HMOG features are collected during taps, we review existing work that uses tap activity to authenticate smartphone users. %
In Table~\ref{tab:relatedResearch}, we summarize the state-of-the-art in tap-based authentication, and highlight various aspects of each work, such as: (1) how the taps were collected---did the user compose free-text or type predefined fixed-text; (2) which body motion conditions (e.g., sitting and walking) were considered; (3) number of subjects (partitioned into owners and impostors, wherever appropriate); (4) how the verifier was trained; (5) how the authentication vector was created; and (6) the features used (e.g., motion-sensor, tap, or keystroke-based). %

Among previous papers~\cite{zheng2012,gascon2014,bo2013}, which have used motion sensors for user authentication, Zheng et al.~\cite{zheng2012} used fixed pins while Gascon et al.~\cite{gascon2014} used fixed phrases. The only work that used free-text typing and also the only one to authenticate users under walking condition is the paper by Bo et al.~\cite{bo2013}. Therefore, we believe that this is closest work to our paper, and highlight the differences between our paper and~\cite{bo2013} as follows: (1) we performed experiments on a large-scale dataset containing 100 users (90 users qualified as genuine, and  93 or more as impostors), while \cite{bo2013} used only 10 genuine users and 50 impostors (on average) from a dataset of 100 subjects. Because the genuine population size in \cite{bo2013} is too small, it is difficult to assess how accurately the reported FARs/FRRs represent the achievable authentication error rates with movement-based features, given that the number of users is a critical factor in assessing the confidence on empirical error rates of biometric systems~\cite{DassZJ06}; 
(2) we introduced and evaluated a wide range of movement features, while~\cite{bo2013} used only two (i.e., mean magnitude of acceleration and mean magnitude of angular velocity, during a gesture). Our results clearly reveal that certain types of movement features (e.g., resistance) perform better than others (e.g., stability), while~\cite{bo2013} does not distinguish between different types of movement features; (3) our evaluation is comprehensive and includes detailed comparison and fusion with additional features such as touchscreen tap and keystroke. This allowed us to report how fusion with different types of features impacted authentication and BKG performance. In contrast,~\cite{bo2013} do not compare different types of features; and (4) HMOG features performed well in both sitting and walking condition, while~\cite{bo2013} had resorted to gait features for authentication during walking.

\paragraph{Biometric Key Generation} 
To our knowledge, there is no previous work on BKG on smartphones. Here, we review some important work related to BKG in general. 

Introduced by Juels et al.~\cite{fuzzy_commitment}, BKG implemented via fuzzy commitments uses error correcting codes to construct cryptographic keys from noisy information. 
Features are extracted from raw signals (e.g., minutiae from fingerprint images); 
then, each feature is encoded using a single bit. Cryptographic keys are {\em committed} using features; subsequently, commitments are opened using biometric signals from the same users. Error-correcting techniques are 
applied to noisy biometric information in order to cope with within-user variance.%

Ballard et al.~\cite{bal08} provided a formal framework for analyzing the security of a BKG scheme, and argued that BKG should enjoy biometric privacy (i.e., biometric signals cannot be reconstructed from biometric keys) and key randomness (i.e., keys look random given their commitment). They also formalized adversarial knowledge of the biometric by introducing \emph{guessing distance}---the logarithm of the number of guesses necessary to open a commitment using feature vectors from multiple impostors.

\paragraph{Energy Consumption Analysis}
Bo et al.~\cite{bo2013} showed that energy consumption can be reduced by selectively turning off motion 
sensors based on two factors: (1) the sensitivity of the app being used---non-sensitive applications, such as games, require no authentication; and (2) the probability that the smartphone is handed to another user. This probability is calculated using historical smartphone usage data. 
In their experiments, Bo et al.~were able to turn off the sensors 30-90\% of the time, while maintaining reasonable authentication performance. However, they did not report how they performed energy consumption measurements, nor listed the  
energy consumptions associated with determining if the phone was being held by its owner or handed to another user. 

Feng et al.~\cite{Feng:2014:TCI:2565585.2565592} introduced TIPS---a continuous user authentication technique that relies on touch features exclusively. By collecting energy usage data, the authors reported average energy consumption of 88 mW, which corresponds to less than 6.2\% overhead. Like~\cite{bo2013},  Feng et al.~\cite{Feng:2014:TCI:2565585.2565592} also do not describe how energy measurements were performed.

Khan et al.~presented Itus \cite{Khan:2014:IIA:2639108.2639141}, a framework that helps Android application developers to deploy various continuous authentication mechanisms. Energy evaluation was performed using PowerTutor~\cite{PowerTutor}---an Android application that reports energy measurements performed by the smartphone. Overall energy overhead of the tested continuous authentication techniques varied between 1.2\% and 6.2\%. 

 Compared to previous research, our work provides a more complete picture of 
energy overhead of continuous authentication using HMOG. In fact, we highlighted 
the tradeoffs of energy usage for different sensor sampling rates and 
authentication scan lengths, versus authentication accuracy. In comparison to our work, existing literature did not analyze fine-grained energy consumption brought by individual components such as motion sensors. 
To our knowledge, we are the first to report 
the relationship between sensor sampling rates and continuous 
authentication accuracy.

\section{Conclusion and Future Work}
\label{sec:conclusion}

In this paper, we introduced HMOG, a set of behavioral biometric features for continuous 
authentication of smartphone users.  We evaluated HMOG from three perspectives---continuous authentication, BKG, and energy consumption. Our evaluation was performed on multi-session data collected from 100 subjects under two motion conditions (i.e., sitting and walking). Results of our evaluation can be summarized as follows. By combining HMOG with tap features, we achieved 8.53\% authentication EER during walking and 11.41\% during sitting, which is lower than the EERs achieved individually with tap or HMOG features. Further, by fusing HMOG, tap and keystroke dynamic features, we achieved the lowest EERs (7.16\% in walking and 10.05\% in sitting). Our results demonstrate that HMOG is well suited for continuous authentication of smartphone users. In fact, HMOG improves the performance of taps and keystroke dynamic features, especially during walking---a common smartphone usage scenario.  For BKG, HMOG features provide lower EER (17.4\%) compared to tap (25.7\%) and swipe features (34.2\%). Moreover, fusion of HMOG with tap features provide the best performance, with 15.1\% EER. Additionally, the energy overhead of sample collection and feature extraction is small (less than 8\% energy overhead when sensors were sampled at 16Hz). This makes HMOG well suited for energy-constrained devices such as smartphones. 
As future work, we plan to investigate how HMOG features perform under stringent constraints such as: (a) walking at higher speeds; (b) using the smartphone in different weather conditions; and (c) using applications that do not involve typing (e.g., browsing a map). Another research question of interest is \textit{cross-device interoperability}, i.e., how and to what extent can a user's behavioral biometric collected on a desktop (e.g., keystroke dynamics) be leveraged with HMOG features to authenticate the user on a smartphone (and vice versa). 

{\footnotesize
 \bibliographystyle{abbrv} 
 \bibliography{references}
}

%
%

%
%
\newpage
\appendices

\section{Security analysis of our BKG Scheme}
\label{sec:security}

We prove that our BKG technique meets the requirements from \cite{bal08}---namely, that 
cryptographic keys are indistinguishable from random 
given the commitment ({\em key randomness}), and that given a cryptographic key and a commitment, no 
useful information about the biometric can be reconstructed ({\em biometric privacy}).
We assume that the biometric is modeled by an unpredictable function. This captures the idea 
that a user's biometric is {\em difficult to guess}. Informally, an 
unpredictable function $f(\cdot)$ is a function for which no efficient adversary 
can predict $f(x^*)$ given $f(x_i)$ for various $x_i \neq x^*$. Formally:

\begin{definition}
A function family $(\C, D, R, F)$ for $\{f_c(\cdot): D \rightarrow~R\}_{c \leftarrow \C}$ is unpredictable if for any efficient algorithm $\A$ and auxiliary information $z$ we have:
\[
	Pr[(x^*, f_c(x^*) \leftarrow \A^{f_c(\cdot)}(z) \text{\; and\; } x^* \not\in Q] \leq \operatorname{negl}(\kappa)
\]
where $Q$ is the set of queries from $\A$, $\kappa$ is the security parameter and $\operatorname{negl}(\cdot)$ is a negligible function.
\label{def:unpredictable}
\end{definition}

In order to define security of biometric key generation systems, Ballard et al.~\cite{bal08} introduced the notions of \emph{Key Randomness} (REQ-KR), \emph{Weak Biometric Privacy} (REQ-WBP) and \emph{Strong Biometric Privacy} (REQ-SBP). 
We formalize the notion of key randomness by defining Experiment $\operatorname{IND-KR}_\A(\kappa)$:

\medskip
\noindent 
{\bf Experiment $\operatorname{IND-KR}_\A(\kappa)$} 

\smallskip
\begin{enumerate}
	\item $\A$ is provided with a challenge $(\PRF_{c_i}(z|1), \delta)$, $k_b$ and $z$, where $k_0 = \PRF_{c_i}(z|0)$ and $k_1 \leftarrow_R \{0,1\}^\kappa$ for a bit $b \leftarrow_R \{0,1\}$, corresponding to user $i$.
	\item $\A$ is allowed to obtain biometric information $x_j$ for arbitrary users $j$ such that $j\neq i$.
	\item $\A$ outputs a bit $b'$ as its guess for $b$. The experiment outputs $1$ if $b = b'$, and $0$ otherwise. 
\end{enumerate}

\begin{definition}
We say that a biometric key generation scheme has the Key Randomness property if there exist a negligible function $\operatorname{negl}(\cdot)$ such that for any PPT $\A$, $\operatorname{Pr}[\operatorname{IND-KR}_\A(\kappa) = 1] \leq 1/2+\operatorname{negl}(\kappa)$.
\end{definition}

\begin{theorem} \label{thm:key_randomness}
Assuming that the $\PRF$ is a pseudo-random function family and that biometric $x = (x_1, ..., x_n)$ is unpredictable, our Fuzzy Commitment scheme has the Key Randomness property.
\label{thm:REQ-KR}
\end{theorem}

\begin{proof}[Proof of Theorem \ref{thm:key_randomness} (Sketch)]
Because $c = x - \delta$, and $x$ is assumed to be unpredictable, $c$ is unpredictable given $\delta$. We now show that any PPT adversary $\A$ that has advantage $1/2+\Delta(\kappa)$ to win the $\operatorname{IND-KR_\A(\kappa)}$ experiment can be used to construct a distinguisher $\D$ that has similar advantage in distinguishing $\PRF$ from a family of uniformly distributed random functions.

$\D$ is given access to oracle $O(\cdot)$ that selects a random codeword $c$ and a random bit $b$, and responds to a query $q$ with random (consistent) values if $b=1$, and with $\PRF_c(q)$ if $b=0$. 
$\D$ selects a random $z$, a codeword $c'$ and a feature vector $x'$, and sets $\delta' = x' - c'$. Then $\D$ sends $\gamma' = (O(z|1), \delta')$ and $O(z|0)$ to $\A$. %

If $b=0$, then pair $(\gamma', \PRF_{c'}(z|0))$ is indistinguishable from $((\PRF_{c}(z|1), \delta), \PRF_{c}(z|0))$, because $\delta$ and $\delta'$ follow the same distribution, $c$ and $c'$ are unpredictable given $\delta$ and thus both $\PRF_{c}(\cdot)$ and $\PRF_{c'}(\cdot)$ are indistinguishable from random. 
If $b=1$, then $O(\cdot)$ is a random oracle, so $(\gamma', O(z))$ is 
indistinguishable from pair $((\PRF_c(z|1), \delta), \PRF_c(z|0))$: $c$ is 
unpredictable given $\delta$,  therefore $\PRF_c(\cdot)$ is indistinguishable from 
random. 
 
Eventually, $\A$ outputs $b'$, and $\D$ outputs  $b'$ as its guess. It is easy to see that $\D$ wins iff $\A$ wins, so $\D$ is correct with probability $1/2+\Delta(\kappa)$. Therefore,  $\Delta(\cdot)$ must be a negligible function.\end{proof}

REQ-WBP states that the adversary learns no useful information about a biometric 
signal from the commitment and the auxiliary information, while REQ-SBP states 
that the adversary learns no useful information about the biometric given 
auxiliary information, the commitment and the key.
For our BKG algorithms, REQ-SBP implies REQ-WBP. In fact, $\PRF_c(z|1)$ (which is part of the commitment) is known to the adversary, and therefore $k=\PRF_c(z|0)$ does not reveal any additional information. %
From the unpredictability of $x$, it follows that the output of $\PRF_c$ does not reveal $c$, so $\PRF_c(z|1)$ and $k$ do not disclose information about $x$.

\begin{IEEEbiography}[{\includegraphics[width=1in,height=1.25in,clip,keepaspectratio]{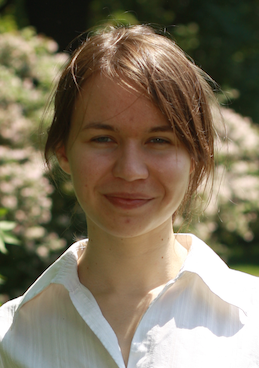}}]{Zde\v nka Sitov\'a} is a bioinformatician at the Mendel Centre for Plant Genomics and Proteomics, Masaryk University, Brno, Czech Republic.
Her research interests include data analysis, machine learning, and their application in life sciences.
She received her B.Sc. and M.Sc. degrees in Computer Science from Masaryk University, Brno, Czech Republic.
\end{IEEEbiography}

\begin{IEEEbiography}[{\includegraphics[width=1in,height=1.2in,clip,keepaspectratio]{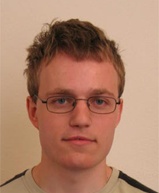}}]{Jaroslav \v Sed\v enka} is pursuing a Ph.D. degree at the Department
of Mathematics and Statistics at Masaryk University, Brno, Czech
Republic. His research interests include lattice cryptography, applied
cryptography, and algebraic number theory.
He received his M.Sc. in Mathematics and B.Sc. in Computer Science
from Masaryk University.

\end{IEEEbiography}

\begin{IEEEbiography}[{\includegraphics[width=1in,height=1.25in,clip,keepaspectratio]{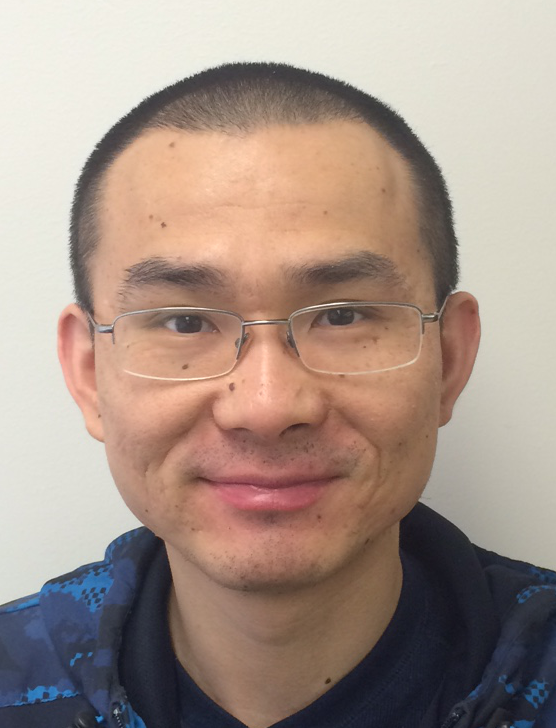}}]{Qing Yang} received his B.S from Civil Aviation University of
China in 2003 and M.S. from Chinese Academy of Sciences in 2007. Since
Fall 2011, he is a Ph.D student in the Department of Computer Science,
College of William \& Mary. His research interests are smartphone
security and energy efficiency.
\end{IEEEbiography}

\begin{IEEEbiography}[{\includegraphics[width=1in,height=1.25in,clip,keepaspectratio]{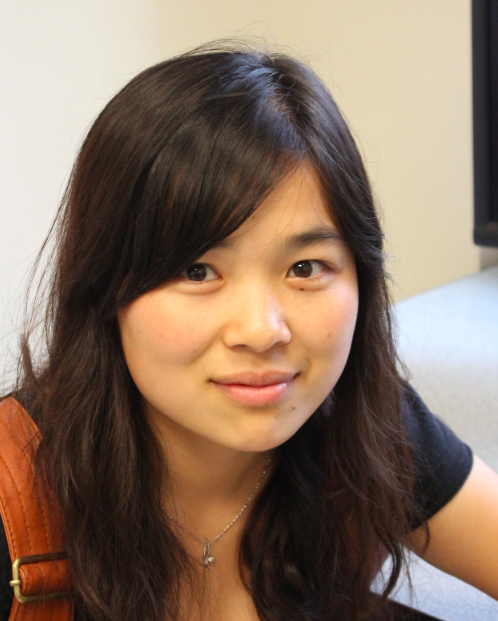}}]{Ge Peng} has been a Ph.D student in the Department of Computer Science, College of William \& Mary, since  Fall 2011. She received her B.S degree from National University of Defense Technology in 2008. Her research interests include wireless networking, smartphone energy efficiency, and ubiquitous computing.
\end{IEEEbiography}

\begin{IEEEbiography}[{\includegraphics[width=1in,height=1.25in,clip,keepaspectratio]{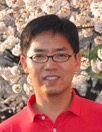}}] {Dr. Gang Zhou} is an Associate Professor in the Computer Science Department at the College of William and Mary. He received his Ph.D. degree from the University of Virginia in 2007. He has published over 70 conference and journal papers in the areas of sensors \& ubiquitous computing, mobile computing, body sensor networks, and wireless networks. The total citations of his papers are more than 4800 according to Google Scholar, among which five of them have been transferred into patents and the MobiSys'04 paper has been cited more than 800 times. He is currently serving in the Journal Editorial Board of IEEE Internet of Things as well as Elsevier Computer Networks. He received an award for his outstanding service to the IEEE Instrumentation and Measurement Society in 2008. He also won the Best Paper Award of IEEE ICNP 2010. He received NSF CAREER Award in 2013. He received a 2015 Plumeri Award for Faculty Excellence. He is a Senior Member of ACM and also a Senior Member of IEEE.
\end{IEEEbiography}

\begin{IEEEbiography}[{\includegraphics[width=1in,height=1.25in,clip,keepaspectratio]{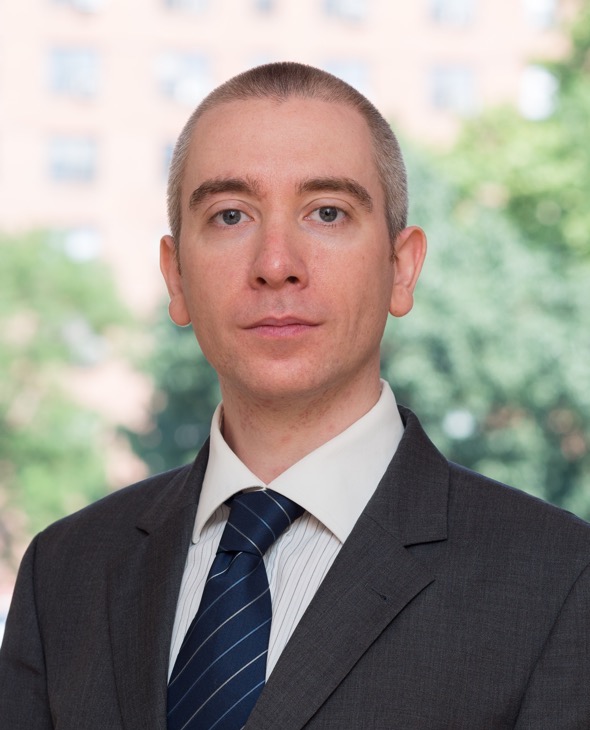}}]{Dr. Paolo Gasti} is an assistant professor of Computer Science at the New York Institute of Technology (NYIT), School of Engineering and Computing Sciences. Dr. Gasti's research focuses on behavioral biometrics, privacy-preserving biometric authentication and identification, secure multi-party protocols and network security. Before joining NYIT, he worked as a research scholar at University of California, Irvine. 
His research has been sponsored by the Defense Advanced Research Project Agency and the U.S. Air Force. He received his B.S., M.S., and Ph.D. degrees from University of Genoa, Italy. He is a Fulbright scholar, and member of IEEE.
\end{IEEEbiography}

\begin{IEEEbiography}[{\includegraphics[width=1in,height=1.25in,clip,keepaspectratio]{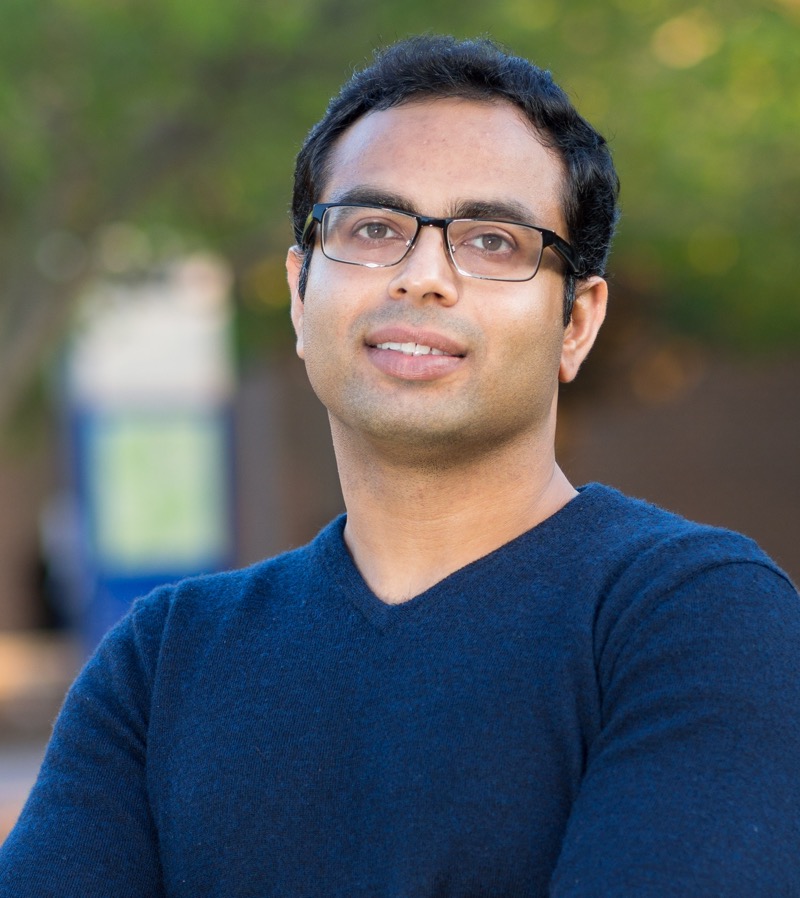}}]{Dr. Kiran S. Balagani} is an assistant professor of computer science at the New York Institute of Technology. His research interests are in cyber-behavioral anomaly detection (e.g., unauthorized user-access behaviors), behavioral biometrics, and privacy-preserving biometrics. Balagani's work has appeared in several peer-reviewed journals, including the IEEE Transactions on Pattern Analysis and Machine Intelligence, the IEEE Transactions on Information Forensics and Security, the IEEE Transactions on Knowledge and Data Engineering, the IEEE Transactions on Systems, Man, and Cybernetics, and Pattern Recognition Letters. He holds three U.S. patents in network-centric attack detection. His teaching interests include development of graduate and undergraduate courses in network security and biometrics. Balagani received the Ph.D. degree and two M.S. degrees from Louisiana Tech University; and the B.S. degree from Bangalore University, India.
\end{IEEEbiography}

\end{document}